\documentclass[12pt]{article}

\usepackage{amsfonts}
\usepackage{amssymb}
\usepackage{amsmath}
\usepackage{amsthm,color}

\usepackage{verbatim}
\usepackage{hyperref}

\newtheorem{thm}{Theorem}
\newtheorem{cor}[thm]{Corollary}
\newtheorem{rem1}[thm]{Remark}
\newtheorem{lemma}[thm]{Lemma}
\newtheorem{defn}[thm]{Definition}
\newtheorem{prop}[thm]{Proposition}

\newtheorem{ex1}[thm]{Example}

\newenvironment{rem}{\begin{rem1}\rm}{\end{rem1}}
\newenvironment{ex}{\begin{ex1}\rm}{\end{ex1}}

\numberwithin{equation}{section}
\numberwithin{thm}{section}

\renewcommand{\P}{\mathbb{P}}
\newcommand{\Q}{\mathbb{Q}}
\newcommand{\R}{\mathbb{R}}
\renewcommand{\S}{\mathbb{S}}
\newcommand{\W}{\mathcal{W}}
\newcommand{\olW}{\W}
\newcommand{\X}{{\bf X}}

\newcommand{\lrparen}[1]{\left(#1\right)}

\newcommand{\lrsquare}[1]{\left[#1\right]}
\newcommand{\lsquare}[1]{\left[#1\right.}
\newcommand{\rsquare}[1]{\left.#1\right]}
\newcommand{\lrcurly}[1]{\left\{#1\right\}}

\newcommand{\Ft}[1]{\mathcal{F}_{#1}}

\newcommand{\Ldp}[2]{L_d^{#1}(#2)}
\newcommand{\Ldq}[1]{\Ldp{q}{#1}}
\newcommand{\Ldpshort}[2]{L^{#1}_{#2}}
\newcommand{\LdpF}[1]{\Ldpshort{p}{#1}}
\newcommand{\LdqF}[1]{\Ldpshort{q}{#1}}
\newcommand{\LdiF}[1]{\Ldpshort{\infty}{#1}}
\newcommand{\LdoF}[1]{\Ldpshort{1}{#1}}
\newcommand{\LdzF}[1]{\Ldpshort{0}{#1}}
\newcommand{\LdpK}[3]{\Ldpshort{#1}{#2}(#3)}

\newcommand{\Lp}[2]{L^{#1}({#2})}
\newcommand{\Lpshort}[2]{L^{#1}_{#2}(\R)}
\newcommand{\LpF}[1]{\Lpshort{p}{#1}}

\newcommand{\LiF}[1]{\Lpshort{\infty}{#1}}
\newcommand{\LoF}[1]{\Lpshort{1}{#1}}

\newcommand{\E}[1]{\mathbb{E}\lrsquare{#1}}
\newcommand{\EQ}[1]{\mathbb{E}^{\mathbb{Q}}\lrsquare{#1}}

\newcommand{\Et}[2]{\E{\left.#1 \right| \mathcal{F}_{#2}}}
\newcommand{\EQt}[2]{\EQ{\left.#1 \right| \mathcal{F}_{#2}}}

\newcommand{\dQdP}{\frac{d\mathbb{Q}}{d\mathbb{P}}}
\newcommand{\dQndP}[1]{\frac{d\mathbb{Q}_{#1}}{d\mathbb{P}}}
\newcommand{\dQidP}{\dQndP{i}}

\newcommand{\dSdP}{\frac{d\mathbb{S}}{d\mathbb{P}}}

\newcommand{\genseq}[4]{(#1_#4)_{#4=#2}^{#3}}
\newcommand{\seq}[1]{\genseq{#1}{0}{T}{t}}

\newcommand{\trans}[1]{#1^{\mathsf{T}}}
\newcommand{\transp}[1]{\trans{\lrparen{#1}}}
\newcommand{\prp}[1]{#1^{\perp}}

\newcommand{\plus}[1]{#1^+}
\newcommand{\plusp}[1]{\plus{(#1)}}

\newcommand{\diag}[1]{#1 }
\newcommand{\cl}{\operatorname{cl}}
\newcommand{\co}{\operatorname{co}}

\newcommand{\interior}{\operatorname{int}}

\newcommand{\env}{\operatorname{env}}

\newcommand{\as}{\text{a.s.}}
\newcommand{\sas}{\text{ a.s.}}
\newcommand{\Pas}{\;\P\text{-}\as}

\DeclareMathOperator*{\esssup}{ess\,sup}
\DeclareMathOperator*{\essinf}{ess\,inf}

\begin{document}

\title{A comparison of techniques for dynamic multivariate risk measures}
\author{Zachary Feinstein \thanks{Washington University in St. Louis, Department of Electrical and Systems Engineering, St. Louis, MO 63130, USA, zfeinstein@ese.wustl.edu.} \and Birgit Rudloff \thanks{Princeton University, Department of Operations Research
and Financial Engineering; and Bendheim Center for Finance, Princeton, NJ 08544, USA,
brudloff@princeton.edu. Research supported by NSF award DMS-1007938.}}
\date{May 9, 2013 (update: January 29, 2015)}
\maketitle
\abstract{This paper contains an overview of results for dynamic multivariate risk measures.  We provide the main results of four different approaches.  We will prove under which assumptions results within these approaches coincide, and how properties like primal and dual representation and time consistency in the different approaches compare to each other.
\\[.2cm]
\\[.2cm]
 }

\section{Introduction}
\label{sec_intro}
The concept of coherent risk measures was introduced in an axiomatic way in \cite{AD97,AD99} to find the minimal capital required to cover the
risk of a portfolio. The notion was relaxed by introducing convex risk measures in \cite{FS02,FS11}.  In these papers the risk was measured only at time zero, in a
frictionless market, for univariate claims, and with only a single eligible asset that can be used for the capital requirements and serves as the num\'{e}raire.  We call this the static scalar framework. In this paper these four assumptions will be removed and different methods compared.

The static assumptions were relaxed by considering dynamic risk measures, where the risk evaluation of a portfolio is updated as time progresses and new information become available.  In the dynamic framework time consistency plays an important role and has been studied for example in \cite{R04,BN04,DS05,RS05,CDK06}.

Eliminating the assumption that the financial markets are frictionless required a new framework.
Since the `value' of a portfolio is not uniquely determined anymore when bid and ask prices or market illiquidity exist, it is natural
to consider portfolios as vectors in physical units instead, i.e. a portfolio is specified by the number of each of the asset which is
held as opposed to their value. But even in the absence of transaction costs multivariate claims might be of interest, e.g. when assets are denoted in different currencies with fluctuating exchange rates, or different business lines with no direct exchange or different regularity rules are considered, see \cite{CM13}. In contrast to frictionless univariate models also the choice of the num\'{e}raire assets matters, which lead to
different approaches: pick a num\'{e}raire and allow capital requirements to be in this num\'{e}raire, which allows a risk manager to work with
scalar risk measures again (see e.g.~\cite{FMM13,ADM09,FS06,K09,Sc04}); or use the more general num\'{e}raire-free approach and allow risk compensation to be made in a basket of assets which leads to risk measures that are set-valued. This approach was first studied in Jouini, Meddeb, Touzi \cite{JMT04} in the coherent case. Several extensions have been made.  In this paper we will introduce four approaches to deal with dynamic multivariate risk measures, and compare and relate them by giving conditions under which the results obtained in each approach coincide.  The four approaches we discuss are
\begin{enumerate}
\item a set-optimization approach;
\item a measurable selector approach;
\item an approach utilizing set-valued portfolios; and
\item a family of multiple asset scalar risk measures.
\end{enumerate}

The first three approaches correspond to the num\'{e}raire-free framework, whereas the last approach
includes scalar risk measures where a num\'{e}raire asset is chosen.

In~\cite{HHH07,HR08,HH10,HHR10} the results of \cite{JMT04} were extended to the convex case and a stochastic market model.  The
extension of the dual representation results were made possible by an application of convex analysis for set-valued functions
(set-optimization), see Hamel \cite{H09}. The dynamic case and time consistency was studied in \cite{FR12,FR12b}. We will call
this approach the set-optimization approach.  The values of risk measures and its minimal elements in this framework have been
studied and computed in \cite{LR11,HRY12,HLR13,LRU13,FR14-alg} via Benson's algorithm for linear vector optimization (see
e.g.~\cite{L11}) in the coherent and polyhedral convex case, respectively via an approximation algorithm in the convex case, see
\cite{LRU13,FR14-alg}.

\cite{TL12} extended the results of \cite{JMT04} for coherent risk measures to the dynamic case. We will call this the measurable selector approach as it considers the value of a risk measures as a random set, and then provides a primal and dual representation for the measurable selectors in that set.  Time consistency properties were also introduced and some equivalent characterizations discussed.

Most recently, in~\cite{CM13}, set-valued coherent risk measures were considered as functions from random sets into the upper sets.  The transaction costs model, and other financial considerations like trading constraints, or illiquidity, are then embedded into the construction of ``set-valued portfolios''. A subclass of risk measures in this framework can be constructed using a vector of scalar risk measures and \cite{CM13} gives upper and lower bounds as well as dual representations for this subclass. We will present here the dynamic extension of this approach. Time consistency properties have not yet been studied within this framework. However, by comparing and relating the different approaches we will see that a larger subclass can be obtained by using the set-valued risk measures of the set-optimization approach, which provides already a link to dual representations and time consistency properties for this larger subclass.

The fourth approach is to consider a family of dynamic scalar risk measures to evaluate the risk of a
multivariate claim. This approach has not been studied so far in the dynamic case.
In the special case of frictionless markets, the family of scalar risk measures coincides with scalar risk measures using multiple eligible assets as discussed in~\cite{FMM13,ADM09,FS06,K09,Sc04}. Also the scalar static risk measure of multivariate claims with a single eligible asset studied in \cite{BR06}; the scalar liquidity adjusted risk measures in market with frictions as studied in \cite{WA13}; and the scalar superhedging price in markets with transaction costs, see \cite{BLPS92,BV92,PL97,JK95,R08,RZ11,LR11}, are special cases of this approach.
Thus, the family of dynamic scalar risk measures for
portfolio vectors generalizes these special cases in a unified way to allow for frictions, multiple eligible assets, and multivariate portfolios in a dynamic framework. The connection to the set-optimization approach allows to utilize the dual representation and time consistency results deduced there.

Other papers in the context of set-valued risk measures are \cite{T06}, where an extension of the tail conditional expectation to the set-valued framework of~\cite{JMT04} was presented and a numerical approximation for calculation was given; and \cite{CM07}, where set-valued risk measures in a more abstract setting were studied and
a consistent structure for scalar-valued, vector-valued, and set-valued risk measures (but for constant solvency cones) was created.  Furthermore, in \cite{CM07} distribution
based risk measures were extended to the set-valued framework via depth-trimmed regions. More recently, vector-valued risk measures were studied in~\cite{BBG12}.

Section~\ref{sec_rm} introduces the four approaches mentioned above.
In section~\ref{sec_relation} these four approaches are compared by showing how the set-optimization  approach corresponds to each of the other three.  For each comparison, assumptions are given under which there is a one-to-one relationship between the approaches.
These relations allow generalizations in most of the different approaches that go beyond the results obtained so far.

\section{Dynamic risk measures}
\label{sec_rm}
Consider a filtered probability space $\lrparen{\Omega,\Ft{},\seq{\mathcal{F}},\P}$ satisfying the usual conditions with $\Ft{0}$ being the completed
trivial sigma algebra and $\Ft{T} = \Ft{}$. Let $|\cdot|$ be an arbitrary norm in $\R^d$.  Denote $\LdpF{t} :=
\Lp{p}{\Omega,\Ft{t},\P;\R^d}$ for $p \in
\lrsquare{0,+\infty}$ (with $\LdpF{} := \LdpF{T}$).  If $p = 0$, $\LdzF{t}$ is the linear space of the
equivalence classes of $\Ft{t}$-measurable functions $X: \Omega \to \R^d$.  For $p > 0$, $\LdpF{t}$ denotes the linear space of
$\Ft{t}$-measurable functions $X: \Omega \to \R^d$ such that $\|X\|_p = \lrparen{\int_{\Omega}|X(\omega)|^p d\P}^{1/p} < +\infty$ for $p \in
(0,+\infty)$, and $\|X\|_{\infty} = \esssup_{\omega \in \Omega} |X(\omega)| < +\infty$ for $p = +\infty$. For $p \in \lrsquare{1,+\infty}$ we will consider the dual pair $\lrparen{\LdpF{t}, \LdqF{t}}$, where $\frac{1}{p}+\frac{1}{q} = 1$ (with $q = +\infty$ when $p = 1$ and $q = 1$ when $p = +\infty$),  and endow it with the norm topology, respectively the weak* topology (that is the $\sigma \lrparen{\LdiF{t},\LdoF{t}}$-topology on $\LdiF{t}$) in the case $p = +\infty$ unless otherwise noted.

We write $\LdpF{t,+} = \lrcurly{X \in \LdpF{t}: \; X \in \R^d_+ \Pas}$ for the closed convex cone of $\R^d$-valued $\Ft{t}$-measurable random vectors with non-negative components.  Similarly define $\LdpF{+} := \LdpF{T,+}$.  We denote by $\LdpK{p}{t}{D_t}$ those random vectors in $\LdpF{t}$ that take $\P$-a.s. values in $D_t$. Let $1_D: \Omega \to \{0,1\}$ be the indicator function of $D \in \Ft{}$ defined by $1_D(\omega) = 1$ if $\omega \in D$ and $0$ otherwise.
Throughout we will consider the summation of sets by Minkowski addition.
To distinguish the spaces of random vectors from those of random variables, we will write $\LpF{t} := \Lp{p}{\Omega,\Ft{t},\P;\R}$ for the linear space of the equivalence classes of $p$ integrable $\Ft{t}$-measurable random variables $X: \Omega \to \R$. Note that an element $X \in \LdpF{t}$ has components $X_1,...,X_d$ in $\LpF{t}$. 

(In-)equalities between random vectors are always understood componentwise in the
$\P$-a.s. sense.  The multiplication between a random variable $\lambda\in \LiF{t}$ and a set of random vectors
$D\subseteq\LdpF{}$ is understood in the elementwise sense, i.e. $\lambda D=\{\lambda Y:Y\in D\}\subseteq\LdpF{}$ with $(\lambda
Y)(\omega)=\lambda(\omega)Y(\omega)$.
The multiplication and division between (random) vectors is understood in the componentwise sense,
i.e. $\diag{x}y := \transp{x_1 y_1,...,x_d y_d}$ and $x/y := \transp{x_1/y_1,...,x_d/y_d}$ for $x,y \in \R^d$ ($x,y \in
\LdpF{t}$) and with $y_i \neq 0$ (almost surely) for every index $i \in \{1,...,d\}$ for division.

As in \cite{K99} and discussed in \cite{S04,KS09}, the portfolios in this paper are in ``physical units'' of an asset rather than the value in a fixed num\'{e}raire, except where otherwise mentioned.  That is, for a portfolio $X \in \LdpF{t}$, the values of $X_i$ (for $1 \leq i \leq d$) are the number of units of asset $i$ in the portfolio at time $t$.

Let $\tilde{M}_t[\omega]$ denote the set of eligible portfolios, i.e. those
portfolios which can be used to compensate for the risk of a portfolio, at time $t$ and state $\omega$. We assume
$\tilde{M}_t[\omega]$ is a linear subspace of $\R^d$ for almost every $\omega \in \Omega$.  It then follows that $M_t :=
\LdpK{p}{t}{\tilde{M}_t}$  is a closed (and additionally weak* closed if $p = +\infty$) linear subspace of $\LdpF{t}$, see section~5.4 and proposition~5.5.1 in \cite{KS09}.
For example, $\tilde{M}_t[\omega]$ could specify a certain ratio of Euros and Dollars to be used for risk compensations. Another typical example is the case where a subset of assets are used for capital requirements, i.e. $\tilde{M}_t^n[\omega] = \lrcurly{m \in \R^d: \forall i \in \{n+1,...,d\}: m_i = 0}$ and $M_t^n = \LdpK{p}{t}{\tilde{M}_t^n}$.
We will denote $M_{t,+} := M_t \cap \LdpF{t,+}$ to be the nonnegative elements of $M_t$.  We will assume that $M_{t,+} \neq \{0\}$,
i.e. $M_{t,+}$ is nontrivial.

In the first three methods discussed below the risk measures have set-valued images.  In the set-optimization  approach (section~\ref{sec_vector}) and the set-valued portfolio approach (section~\ref{sec_fun-sets}) the image space is explicitly given by the upper sets, i.e. $\mathcal{P}\lrparen{M_t;M_{t,+}}$ where $\mathcal{P}\lrparen{\mathcal{Z};C} := \lrcurly{D \subseteq \mathcal{Z}: D = D + C}$ for some vector space $\mathcal{Z}$ and an ordering cone $C \subset \mathcal{Z}$.  Additionally, let $\mathcal{G}(\mathcal{Z};C) := \lrcurly{D \subseteq \mathcal{Z}: D = \cl\co\lrparen{D + C}} \subseteq \mathcal{P}(\mathcal{Z};C)$ be the upper closed convex subsets. It seems natural to use upper set as the values of risk measures since if one portfolio can cover the risk then any larger portfolio should also cover this risk.  Alternatively, one could consider the set of ``minimal elements'' of the risk compensating portfolios. However, in contrast to the upper sets, the  set of ``minimal elements'' is in general not a convex set when convex risk measure are considered.

\subsection{Set-optimization approach}
\label{sec_vector}

The set-optimization approach to dynamic risk measures is studied in \cite{FR12,FR12b}, where set-valued risk measures (\cite{HH10,HHR10}) were extended to the dynamic case.
A benefit of this method is that dual representations are obtained by a direct application of
the set-valued duality developed in~\cite{H09}, which allowed for the first time to study not only conditional coherent, but also convex set-valued risk measures.

In this setting we consider risk measures that map a portfolio vector into the
complete lattice $\mathcal{P}\lrparen{M_t;M_{t,+}}$ of upper sets.

Set-valued conditional risk measures have been defined in \cite{FR12}.  Here we give a stronger property for finiteness at zero than in \cite{FR12} to ease the comparison to the other approaches.
\begin{defn}
\label{defn_conditional}
A \textbf{\emph{conditional risk measure}} is a mapping $R_t: \LdpF{} \to \mathcal{P}(M_t;M_{t,+})$
which satisfies:
\begin{enumerate}
\item $\LdpF{+}$-monotonicity: if $Y - X \in \LdpF{+}$ then $R_t(Y) \supseteq R_t(X)$;
\item $M_t$-translativity: $R_t(X+m) = R_t(X)-m$ for any $X \in \LdpF{}$ and $m \in M_t$;
\item finiteness at zero: $R_t(0) \neq\emptyset$ and $R_t(0)[\omega] \neq \tilde{M}_t[\omega]$ for almost every $\omega\in\Omega$, where $R_t(0)[\omega] := \lrcurly{u(\omega): u \in R_t(0)}$.
\end{enumerate}
\end{defn}
For finiteness at zero, and elsewhere in later sections, we consider the $\omega$ projection of the risk compensating set $R_t(X)$.  We point out that $R_t(X)$ is a collection of random vectors and is \emph{not} a random set; therefore $R_t(X)[\omega] := \lrcurly{u(\omega): u \in R_t(X)}$ is the collection of risk covering portfolios at state $\omega$.  As $R_t(X)$ is not a random set, it is generally the case that $R_t(X) \neq \LdpK{p}{t}{R_t(X)} := \lrcurly{u \in M_t: \P\lrparen{\omega \in \Omega: u(\omega) \in R_t(X)[\omega]} = 1}$.

Below we consider additional properties for conditional risk measures that have useful financial and mathematical interpretations.  Note that the definition for $K$-compatibility below is more general than the one given in~\cite{FR12}, and corresponds to the definition in~\cite{HRY12}.
A conditional risk measure $R_t$ at time $t$ is
\begin{itemize}
\item\textbf{\emph{convex (conditionally convex)}} if for all $X,Y \in \LdpF{}$ and any $\lambda \in [0,1]$ (respectively $\lambda \in \LiF{t}$ such that $0 \leq \lambda \leq 1$)
    \[R_t(\lambda X + (1-\lambda)Y) \supseteq \lambda R_t(X) + (1-\lambda) R_t(Y);\]

\item\textbf{\emph{positive homogeneous (conditionally positive homogeneous)}} if for all $X \in \LdpF{}$ and any $\lambda \in \R_{++}$ (respectively $\lambda \in L^\infty_t(\R_{++})$)
    \[R_t(\lambda X) = \lambda R_t(X);\]

\item\textbf{\emph{coherent (conditionally coherent)}} if it is convex and positive homogeneous (respectively conditionally convex and conditionally positive homogeneous);

\item \textbf{\emph{normalized}} if $R_t(X) = R_t(X) + R_t(0)$ for every $X \in \LdpF{}$;

\item \textbf{\emph{local}} if for every $D \in \Ft{t}$ and every $X \in \LdpF{}$, $1_D R_t(X) = 1_D R_t(1_D X)$;

\item\textbf{\emph{$K$-compatible}} for some convex cone $K \subseteq \LdpF{}$ if $R_t(X) = \bigcup_{k \in K} R_t(X- k)$;

\item \textbf{\emph{closed}} if the graph of the risk measure
    \[\operatorname{graph}(R_t) = \lrcurly{(X,u) \in \LdpF{} \times M_t: u \in R_t(X)}\]
    is closed in the product topology (with the weak* topology if $p = +\infty$);

\item\textbf{\emph{convex upper continuous}} if
    \[R_t^{-1}(D) := \lrcurly{X \in \LdpF{}: R_t(X) \cap D \neq \emptyset}\]
    is closed (weak* closed if $p = +\infty$) for any closed convex set $D \in \mathcal{G}(M_t;M_{t,-})$.
\end{itemize}

(Conditional) convexity and coherence for a risk measure define a regulatory framework which promotes diversification.  Set-valued normalization is a generalization of the scalar normalization (zero capital needed to compensate the risk of the $0$ portfolio).  The local property means that the risks at some state (in $\Ft{t}$) only depend on the possible future values of the portfolio reachable from that state.  $K$-compatibility is closely related to a market model; assume for the moment an investor can trade the initial portfolio $0$ into any random vector in $-K$ by the terminal time $T$, then $K$-compatibility means considering the (minimal) risk of a portfolio when all possible trades  are taken into account.
The closure is the set-valued version of lower semicontinuity and is necessary for the dual representation to hold. Convex upper continuity is a stronger property than closure and is useful when characterizing or creating multi-portfolio time consistent risk measures, the details will be given below.

A \textbf{\emph{dynamic risk measure}} is a sequence $\seq{R}$ of conditional risk measures.  A dynamic risk measure is said to have a certain property if $R_t$ has that property for all times $t$.

A static risk measure in the sense of~\cite{HHR10} is a conditional risk measure at time $0$.
Note that for static risk measures convexity (positive homogeneity) coincides with conditional convexity (conditional positive homogeneity).

Any conditionally convex risk measure $R_t: \LdpF{} \to \mathcal{P}\lrparen{M_t;M_{t,+}}$ is local, see proposition 2.8 in~\cite{FR12}.

\begin{defn}
\label{defn_acceptance}
A set $A_t \subseteq \LdpF{}$ is a \textbf{\emph{conditional acceptance set}} at time $t$ if it satisfies $A_t + \LdpF{+} \subseteq A_t$,
$M_t \cap A_t \neq \emptyset$, and $\tilde{M}_t[\omega] \cap (\R^d \backslash A_t[\omega]) \neq \emptyset$ for almost every $\omega \in \Omega$, where $A_t[\omega] = \{X(\omega): X \in A_t\}$.
\end{defn}

The acceptance set of a conditional risk measure $R_t$ is given by $A_t = \lrcurly{X \in \LdpF{}: 0 \in
R_t(X)}$, which is the collection of ``risk free'' portfolios.  For any conditional acceptance set $A_t$,  the function defined by $R_t(X) = \lrcurly{u \in
M_t: X + u \in A_t}$ is a conditional risk measure. This is the primal representation for conditional risk measures via acceptance sets, see~\cite{FR12}.  This relation is one-to-one, i.e. we can consider an $(R_t,A_t)$ pair
or equivalently just one of the two.
Given a risk measure and acceptance set pair $(R_t,A_t)$ then the following properties hold, see Proposition 2.11 in~\cite{FR12}.
\begin{itemize}
\item $R_t$ is normalized if and only if $A_t + A_t \cap M_t = A_t$;
\item $R_t$ is (conditionally) convex if and only if $A_t$ is (conditionally) convex;
\item $R_t$ is (conditionally) positive homogeneous if and only if $A_t$ is a (conditional) cone;
\item $R_t$ has a closed graph if and only if $A_t$ is closed.
\end{itemize}

For the duality results below we will consider $p \in [1,+\infty]$.
Let $\mathcal{M}$ denote the set of $d$-dimensional probability measures absolutely continuous with respect to $\P$, and let $\mathcal{M}^e$ denote the set of $d$-dimensional probability measures equivalent to $\P$.
We will say $\Q = \P|_{\Ft{t}}$ for vector probability measures $\Q$ and some time $t\in [0,T]$ if for every $D \in\Ft{t}$ it follows that $\Q_i(D) = \P(D)$ for all $i = 1,...,d$.
Consider $\Q \in \mathcal{M}$. We will use a $\P$-almost sure version of the $\Q$-conditional expectation of $X \in \LdpF{}$ given by
\[\EQt{X}{t} := \Et{\diag{\xi_{t,T}(\Q)} X}{t},\] where $\xi_{r,s}(\Q) =
\transp{\bar{\xi}_{r,s}(\Q_1),...,\bar{\xi}_{r,s}(\Q_d)}$ for any times $0 \leq r \leq s \leq T$ with \[\bar{\xi}_{r,s}(\Q_i)[\omega] :=
\begin{cases}\frac{\Et{\dQidP}{s}(\omega)}{\Et{\dQidP}{r}(\omega)} & \text{on } \Et{\dQidP}{r}(\omega) > 0\\ 1 & \text{else}
\end{cases}\] for every $\omega \in \Omega$, see e.g.~\cite{CK10,FR12}.  For any probability measure $\Q_i \ll \P$ and any times $0 \leq r \leq s \leq t \leq T$, it follows that $\dQidP = \bar{\xi}_{0,T}(\Q_i)$,
$\bar{\xi}_{t,s}(\Q_i) = \bar{\xi}_{t,r}(\Q_i) \bar{\xi}_{r,s}(\Q_i)$, and $\Et{\bar{\xi}_{r,s}(\Q_i)}{r} = 1$ almost
surely.
The halfspace and the conditional ``halfspace'' in $\LdpF{t}$ with normal direction $w \in \LdqF{t}\backslash \{0\}$ are denoted by \[G_t(w) := \lrcurly{u \in \LdpF{t}: 0 \leq \E{\trans{w}u}}, \quad\quad \Gamma_t(w) := \lrcurly{u \in \LdpF{t}: 0 \leq \trans{w}u \Pas}.\]
We will define the set of dual variables to be
\begin{equation*}
\W_t := \lrcurly{(\Q,w) \in \mathcal{M} \times \lrparen{\plus{M_{t,+}} \backslash \prp{M_t}}: w_t^T(\Q,w) \in \LdqF{+}, \Q = \P|_{\Ft{t}}},
\end{equation*}
where for any $0 \leq t \leq s\leq T$
\[w_t^{s}(\Q,w) = \diag{w}\xi_{t,s}(\Q),\]
$\prp{M_t} = \lrcurly{v \in \LdqF{t}: \E{\trans{v}u} = 0 \; \forall u \in M_t}$ and $C^+=\lrcurly{v \in \LdqF{t}: \E{\trans{v}u} \geq 0 \; \forall u \in C}$ denotes the positive dual cone of a cone $C\subseteq \LdpF{t}$.

The set of dual variables $\W_t$ consists of two elements.  The first component is a vector probability measure absolutely continuous to the physical measure $\P$ and corresponds to the dual element in the traditional scalar theory.
The second component reflects the order relation in the image space as the $w$'s are the collection of possible relative weights between the eligible portfolios. This component is not needed in the scalar case.
The coupling condition $w_t^T(\Q,w) \in \LdqF{+}$ guarantees that the probability measure $\Q$ and the ordering vector $w$ are ``consistent'' in the following sense.
If a portfolio $X$ is component-wise ($\P$-)almost surely greater than or equal to another portfolio $Y$, then the $\Q$-conditional expectation keeps that relationship with respect to the order relation defined by $w$, that is $\trans{w}\EQt{X}{t}\geq\trans{w}\EQt{Y}{t}$ ($\P$-)almost surely.
In the following, we review the duality results from~\cite{FR12b}.  Note that since we are only considering closed (conditionally) convex risk measures we can restrict the image space to $\mathcal{G}(M_t;M_{t,+}) := \lrcurly{D \subseteq M_t: D = \cl\co\lrparen{D + M_{t,+}}}$.

\begin{cor}[Corollary 2.4 of~\cite{FR12b}]
\label{cor_conditional_dual}
A conditional risk measure $R_t: \LdpF{} \to \mathcal{G}(M_t;M_{t,+})$ is closed and conditionally convex if and only if
\begin{equation}
\label{conditional_convex_dual}
R_t(X) = \bigcap_{(\Q,w) \in \olW_t} \lrsquare{-\alpha_t^{\min}(\Q,w) + \lrparen{\EQt{-X}{t} + \Gamma_t\lrparen{w}} \cap M_t},
\end{equation}
where $-\alpha_t^{\min}$ is the minimal conditional penalty function given by
\begin{equation}
\label{conditional_min penalty}
-\alpha_t^{\min}(\Q,w) = \operatorname{cl}\bigcup_{Z \in A_t} \lrparen{\EQt{Z}{t} + \Gamma_t(w)} \cap M_t.
\end{equation}
$R_t$ is additionally conditionally coherent if and only if
\begin{equation}
\label{conditional_coherent_dual}
R_t(X) = \bigcap_{(\Q,w) \in \olW_{t}^{\max}} \lrparen{\EQt{-X}{t} + \Gamma_t\lrparen{w}} \cap M_t
\end{equation}
with
\begin{equation}
\label{max_dualset}
\W_{t}^{\max} = \lrcurly{(\Q,w) \in \W_t: w_t^T(\Q,w) \in \plus{A_t}}.
\end{equation}
\end{cor}

The more general convex and coherent case reads analogously to corollary~\ref{cor_conditional_dual}, just with
$\Gamma_t\lrparen{w}$ replaced by $G_t\lrparen{w}$ in equations~\eqref{conditional_convex_dual},~\eqref{conditional_min penalty}
and \eqref{conditional_coherent_dual}, see theorem~2.3 in~\cite{FR12b}. As shown in~\cite{HHR10,FR12}, the $G_t$-version of the minimal
penalty function $-\alpha_t^{\min}$ is the set-valued (negative) convex conjugate in the sense of~\cite{H09} and
the dual representation is the biconjugate, both with infimum and supremum defined for the image space
$\mathcal{G}(M_t;M_{t,+})$.
\begin{rem}
The dual representation given in~\cite{JMT04} for the static coherent case (and in~\cite{TL12} for the dynamic case, see
section~\ref{sec_measurable} below) uses a single dual variable.  This set of dual variables from~\cite{JMT04} is equivalent to
$\lrcurly{w_t^T(\Q,w): (\Q,w) \in \W_t}$, and as discussed in the proof of theorem 2.3 in~\cite{FR12b}, the dual
representation~\eqref{conditional_convex_dual} and~\eqref{conditional_coherent_dual} can be given by this set alone.  This means, the results presented in this section include the previously known dual representation results.
\end{rem}

We conclude this section by giving a brief description and equivalent characterizations of a time consistency property for set-valued risk measures in the set-optimization  approach.  The property we will discuss is multi-portfolio time consistency, which was proposed in~\cite{FR12} and further studied in~\cite{FR12b}.  We also return to the general case with $p \in [0,+\infty]$.
\begin{defn}
\label{defn_mptc}
A dynamic risk measure $\seq{R}$ is called \textbf{\emph{multi-portfolio time consistent}} if for all times $t,s \in [0,T]$ with $t < s$, all portfolios $X\in \LdpF{}$ and all sets ${\bf Y}\subseteq \LdpF{}$ the implication
\begin{equation}
  R_{s}(X) \subseteq \bigcup_{Y \in {\bf Y}} R_{s}(Y) \Rightarrow R_t(X) \subseteq \bigcup_{Y \in {\bf Y}} R_t(Y)
\end{equation}
is satisfied.
\end{defn}
Multi-portfolio time consistency means that if at some time any risk compensation portfolio for $X$ also compensates the risk of some portfolio $Y$ in the set ${\bf Y}$, then at any prior time the same relation should hold true. Implicitly within the definition, the choice of eligible portfolios can have an impact on the multi-portfolio time consistency of a risk measure.

In~\cite{FR12}, (set-valued) time consistency was also introduced.  This property is defined by
\[R_{s}(X) \subseteq R_{s}(Y) \Rightarrow R_t(X) \subseteq R_t(Y)\]
for any time $t,s \in [0,T]$ with $t < s$ and any portfolios $X,Y \in \LdpF{}$.  It is weaker than multi-portfolio time consistency, though in the scalar case both properties coincide.

Before we give some equivalent characterizations for multi-portfolio time consistency, we must give a few additional definitions.
These definitions are used for defining the stepped risk measures $R_{t,s}: M_{s} \to \mathcal{P}(M_t;M_{t,+})$ for $t \leq s$, as
discussed in~\cite[appendix C]{FR12b}.  We denote and define the stepped acceptance set by $A_{t,s} := A_t \cap M_{s}$.  And akin to corollary~\ref{cor_conditional_dual}, for the closed conditionally convex and closed (conditionally) coherent stepped risk measures we will define the minimal stepped penalty function (for the conditionally convex case with $M_t \subseteq M_{s}$) by $-\alpha_{t,s}^{\min}(\Q,w) := \cl \bigcup_{X \in A_{t,s}} \lrparen{\EQt{X}{t} + \Gamma_t(w)} \cap M_t$ for every $(\Q,w) \in \W_{t,s}$ and the maximal stepped dual set (for the  (conditionally) coherent case with $M_t \subseteq M_{s}$) by $\W_{t,s}^{\max} := \lrcurly{(\Q,w) \in \W_{t,s}: w_t^{s}(\Q,w) \in \plus{A_{t,s}}}$.  As can be seen, both the stepped penalty function and the stepped maximal dual set are with respect to dual elements $\W_{t,s}$, which in general differ from $\W_t$.  In the case that $\tilde{M}_t = M_0$ almost surely, it holds $\W_{t,s} \supseteq \W_t$ for all times $t \leq s \leq T$; if $\tilde{M}_{s} = \R^d$ almost surely then $\W_{t,s} = \W_t$.

In the below theorem, for the convex upper continuous (conditionally) coherent case we introduce two more definitions.  We define the mapping $H_t^{s}: 2^{\W_{s}} \to 2^{\W_t}$ for times $t \leq s \leq T$ by $H_t^{s}(D) := \lrcurly{(\Q,w) \in \W_t: \lrparen{\Q,w_t^{s}(\Q,w)} \in D}$ for $D \subseteq \W_{s}$.  Additionally, for $\Q,\R \in \mathcal{M}$ we denote by $\Q \oplus^{s} \R$ the pasting of $\Q$ and $\R$ in $s$, i.e. the vector probability measures $\S\in \mathcal{M}$ defined via
\[\dSdP = \diag{\xi_{0,s}(\Q)} \xi_{s,T}(\R).\]

The following theorem gives equivalent characterizations of multi-portfolio time consistency: a recursion in the spirit of Bellman's principle (property \ref{thm_equiv_recursive}  below), an additive property for the acceptance sets (property \ref{thm_equiv_acceptance}), the so called cocyclical property (property \ref{thm_equiv_conditional_penalty}) and stability (property \ref{thm_equiv_stable}). The properties are important for the construction of  multi-portfolio time consistent risk measures.
\begin{thm}[Theorem 3.4 of \cite{FR12}, corollary 3.5, corollary 4.3 and theorem 4.6 of \cite{FR12b}] 
\label{thm_mptc_equiv}
For a normalized dynamic risk measure $\seq{R}$ the following are equivalent:
\begin{enumerate}
\item \label{thm_equiv_mptc}$\seq{R}$ is multi-portfolio time consistent,
\item \label{thm_equiv_recursive} $R_t$ is recursive, that is for every time $t,s \in [0,T]$ with $t < s$
    \begin{equation}
    \label{recursive}
        R_t(X) = \bigcup_{Z \in R_{s}(X)} R_t(-Z) =: R_t(-R_{s}(X)).
    \end{equation}
\end{enumerate}
If additionally $M_t \subseteq M_{s}$ for every time $t,s \in [0,T]$ with $t < s$ then all of the above is also equivalent to
\begin{enumerate}
\setcounter{enumi}{2}
\item \label{thm_equiv_acceptance} for every time $t,s \in [0,T]$ with $t < s$
    \begin{equation}
    \label{sum_acceptance}
        A_t = A_{s} + A_{t,s}.
    \end{equation}
\end{enumerate}
If additionally $p \in [1,+\infty]$, $\tilde{M}_t = \R^n \times \{0\}^{d-n}$ almost surely for some $n \leq d$ for every time $t \in
[0,T]$, $\seq{R}$ is a c.u.c. conditionally convex risk measure and
\begin{equation*}
R_t(X) = \bigcap_{(\Q,w) \in \olW_t^e} \lrsquare{-\alpha_t^{\min}(\Q,w) + \lrparen{\EQt{-X}{t} + \Gamma_t(w)} \cap M_t}
\end{equation*}
for every $X \in \LdpF{T}$ where $\olW_t^e = \lrcurly{(\Q,w) \in \olW_t: \Q \in \mathcal{M}^e}$,
then all of the above is also equivalent to
\begin{enumerate}
\setcounter{enumi}{3}
\item \label{thm_equiv_conditional_penalty} for every time $t,s \in [0,T]$ with $t < s$
    \begin{equation}
    \label{sum_conditional_penalty}
        -\alpha_t^{\min}(\Q,w) = \cl \lrparen{-\alpha_{t,s}^{\min}(\Q,w) + \EQt{-\alpha_{s}^{\min}(\Q,w_t^{s}(\Q,w))}{t}}
    \end{equation}
    for every $(\Q,w) \in \W_t^e$.
\end{enumerate}
If additionally $p \in [1,+\infty]$, $\tilde{M}_t = \R^n \times \{0\}^{d-n}$ almost surely for some $n \leq d$ for every time $t \in [0,T]$ and $\seq{R}$ is a c.u.c. (conditionally) coherent risk measure then all of the above is also equivalent to
\begin{enumerate}
\setcounter{enumi}{4}
\item \label{thm_equiv_compose} for every time $t,s \in [0,T]$ with $t < s$
    \begin{equation}
    \label{intersect_dual}
        \W_t^{\max} = \W_{t,s}^{\max} \cap H_t^{s}\lrparen{\W_{s}^{\max}},
    \end{equation}
    which in turn is equivalent to
\item \label{thm_equiv_stable} for every time $t,s \in [0,T]$ with $t < s$
    \begin{equation}
    \label{eq_stable}
    \olW_t^{\max} = \lrcurly{(\Q \oplus^{s} \R,w): (\Q,w) \in \olW_{t,s}^{\max}, (\R,w_t^{s}(\Q,w)) \in \olW_{s}^{\max}}.
    \end{equation}
\end{enumerate}
\end{thm}

\subsection{Measurable selector approach}
\label{sec_measurable}

The measurable selector approach was proposed in~\cite{TL12}  and is an extension of~\cite{JMT04} to the dynamic framework. Only coherent risk measures are considered in this approach
as the technique used to deduce the dual representation relies on coherency.
The risk measures are assumed to be compatible to a conical market model at the final time $T$, i.e. portfolios are compared
based on the final ``values''.  In so doing, a new pre-image space denoted by $B_{K_T,n}$ is introduced, which will be defined below and is discuss in
remark~\ref{rem_preimage}.
In~\cite{TL12}, the space of eligible assets is $M_t^n = \LdpK{0}{t}{\tilde{M}_t^n}$ with $\tilde{M}_t^n[\omega] =
\lrcurly{m \in \R^d: \forall i \in \{n+1,...,d\}: m_i = 0}$, i.e. $n \leq d$ of the $d$ assets can be used to cover risk.

Let $\mathcal{S}_t^d$ be the set of $\Ft{t}$-measurable random sets in $\R^d$.  Recall that a mapping $\Gamma: \Omega \to 2^{\R^d}$ is an $\Ft{t}$-measurable random set if
\[\operatorname{graph} \Gamma = \lrcurly{(\omega,x) \in \Omega \times \R^d: x \in \Gamma(\omega)}\]
is $\Ft{t} \otimes \mathcal{B}(\R^d)$-measurable (where $\mathcal{B}(\R^d)$ is the Borel $\sigma$-algebra).  The random set $\Gamma$ is closed (convex, conical) if $\Gamma(\omega)$ is closed (convex, conical) for almost every $\omega \in \Omega$.

Let $K_T \in \mathcal{S}_T^d$ satisfy the following assumptions:
\renewcommand{\theenumi}{k\arabic{enumi}}
\begin{enumerate}
\item for almost every $\omega \in \Omega$: $K_T(\omega)$ is a closed convex cone in $\R^d$;
\item for almost every $\omega \in \Omega$: $\R^d_+ \subseteq K_T(\omega) \neq \R^d$;
\item for almost every $\omega \in \Omega$: $K_T(\omega)$ is a proper cone, i.e. $K_T(\omega) \cap -K_T(\omega) = \{0\}$.
\end{enumerate}
\renewcommand{\theenumi}{\arabic{enumi}}
It is then possible to create a partial ordering in $\LdzF{}$ defined by $K_T$ such that $X \geq_{K_T} Y$ if and only if $\P(X - Y
\in K_T) = 1$. The solvency cones with friction, see e.g.~\cite{K99,S04,KS09}, satisfy the conditions given above for $K_T$.

Let $n \leq d$, then we define $B_{K_T,n} := \lrcurly{X \in \LdzF{}: \exists c \in \R_+: c 1_{d,n} \geq_{K_T} X \geq_{K_T} -c
1_{d,n}}$ where the $i$-th component of $1_{d,n}\in\R^d$ is $1_{d,n}^i = \begin{cases}1 &\text{if }i \in \{1,...,n\}\\ 0 &\text{else}\end{cases}$.  Then we can define a norm on
$B_{K_T,n}$ by $\|X\|_{K_T,n} := \inf\lrcurly{c \in \R_+: c 1_{d,n} \geq_{K_T} X \geq_{K_T} -c 1_{d,n}}$, and
$(B_{K_T,n},\|\cdot\|_{K_T,n})$ defines a Banach space.

Let $\mathcal{S}_t^{d,n} \subseteq \mathcal{S}_t^d$
be such that $\Gamma \in
\mathcal{S}_t^{d,n}$ if  $\Gamma \in\mathcal{S}_t^d$ and $\Gamma(\omega) \subseteq \tilde{M}_t^n[\omega]$ for almost every $\omega \in \Omega$.

\begin{defn}
\label{defn_risk_process}
A \textbf{\emph{risk process}} is a sequence $\seq{\tilde{R}}$ of mappings $\tilde{R}_t: B_{K_T,n} \to \mathcal{S}_t^{d,n}$ satisfying
\begin{enumerate}
\item $\tilde{R}_t(X)$ is a closed $\Ft{t}$-measurable random set for any $X \in B_{K_T,n}$, $\tilde{R}_t(0) \neq \emptyset$, and
$\tilde{R}_t(0)[\omega] \neq \tilde{M}_t^n[\omega]$  for almost every $\omega \in \Omega$.
\item For any $X,Y \in B_{K_T,n}$ with $Y \geq_{K_T} X$ it holds $\tilde{R}_t(Y) \supseteq \tilde{R}_t(X)$.
\item $\tilde{R}_t(X + m) = \tilde{R}_t(X)-m$ for any $X \in B_{K_T,n}$ and $m \in M_t^n$.
\end{enumerate}

A risk process is \textbf{\emph{conditionally convex}} at time $t$ if for all $X,Y \in B_{K_T,n}$ and $\lambda \in \LiF{t}$ with $0 \leq \lambda \leq 1$ almost surely it holds $\lambda \tilde{R}_t(X) + (1-\lambda) \tilde{R}_t(Y)
\subseteq \tilde{R}_t(\lambda X + (1-\lambda) Y)$.

A risk process is \textbf{\emph{conditionally positive homogeneous}} at time $t$ if for all $X \in B_{K_T,n}$ and $\lambda \in L^0_t(\R_{++})$ with $\lambda X \in B_{K_T,n}$ it holds $\tilde{R}_t(\lambda X) = \lambda \tilde{R}_t(X)$.

A risk process is \textbf{\emph{conditionally coherent}} at time $t$ if it is both conditionally convex and conditionally positive homogeneous at time $t$.

A risk process is \textbf{\emph{normalized}} at time $t$ if $\tilde{R}_t(X) + \tilde{R}_t(0) = \tilde{R}_t(X)$ for every $X \in B_{K_T,n}$.
\end{defn}

Thus, the values $\tilde{R}_t(X)$ of a risk process are $\Ft{t}$-measurable random sets in $\R^d$. Primal and dual representations can be provided for the measurable selectors of this set.
Recall that $\gamma$ is a $\Ft{t}$-measurable selector of a $\Ft{t}$-random set $\Gamma$ if $\gamma(\omega) \in \Gamma(\omega)$ for almost
every $\omega \in \Omega$.  Then using the notation from above, the measurable selectors in $L^p$ are given by $\LdpK{p}{t}{\Gamma} =
\lrcurly{\gamma \in \LdpF{t}: \P(\gamma \in \Gamma) = 1}$.

\begin{defn}
\label{defn_risks_selector}
Given a risk process $\seq{\tilde{R}}$, then $S_{\tilde{R}}: [0,T] \times B_{K_T,n} \to 2^{M_t^n}$ is a \textbf{\emph{selector risk measure}} if $S_{\tilde{R}}(t,X) := \LdpK{0}{t}{\tilde{R}_t(X)}$ for every time $t$ and portfolio $X \in B_{K_T,n}$.  The \textbf{\emph{bounded selector risk measure}} is defined by $S_{\tilde{R}}^{\infty}(t,X) := S_{\tilde{R}}(t,X) \cap B_{K_T,n}$.
\end{defn}

\begin{defn}
\label{defn_msb-acceptance}
A set $A_t \subseteq B_{K_T,n}$ is a \textbf{\emph{conditional acceptance set}} at time $t$ if:
\begin{enumerate}
\item $A_t$ is closed in the $(B_{K_T,n},\|\cdot\|_{K_T,n})$ topology.
\item If $X \in B_{K_T,n}$ such that $X \geq_{K_T} 0$ then $X \in A_t$.
\item $B_{K_T,n}\cap M_t^n \not\subseteq A_t$.
\item $A_t$ is $\Ft{t}$-decomposable, i.e. if for any finite partition $(\Omega_t^n)_{n =1}^N \subseteq \Ft{t}$ of
$\Omega$ and any family $(X_n)_{n =1}^N \subseteq D$, then $\sum_{n =1}^N 1_{\Omega_t^n} X_n \in D$.
\item $A_t$ is a conditionally convex cone.
\end{enumerate}
\end{defn}

\begin{rem}
\label{rem_process-acceptance}
Note that the definition for $\Ft{t}$-decomposability above differs from that in~\cite{TL12}, as in that paper $\Ft{t}$-decomposability is considered with respect to countable rather than finite partitions.  We weakened the condition by adapting the proof of theorem 1.6 of chapter 2 from~\cite{M05} when $p = +\infty$ to the space $B_{K_T,n}$.
\end{rem}

\begin{prop}[Proposition 3.4 of~\cite{TL12}]
\label{prop_process-acceptance}
Given a conditionally coherent risk process $\tilde{R}_t$ at time $t$, then $A_t := \lrcurly{X \in B_{K_T,n}: 0 \in \tilde{R}_t(X)}$ is a conditional acceptance set at time $t$.
\end{prop}

A primal representation of the selector risk measure is given as follows.
\begin{thm}[Theorem 3.3 in \cite{TL12}]
\label{thm_process-primal}
Let $A_t$ be a closed subset of $(B_{K_T,n},\|\cdot\|_{K_T,n})$.  Then $A_t$ is a conditional acceptance set if and only if there
exists some conditionally coherent risk process $\tilde{R}_t$ at time $t$ such that the associated bounded selector risk measure $S_{\tilde{R}}^{\infty}$ satisfies $S_{\tilde{R}}^{\infty}(t,X) = \lrcurly{m \in M_t^n: X + m \in A_t}$ for all $X\in B_{K_T,n}$.
\end{thm}

Below, we give the dual representation for coherent selector risk measures as done in theorem~4.1 and theorem~4.2
of~\cite{TL12}.  This dual representation can be viewed as the intersection of supporting halfspaces for the selector risk measure, which is the reason that coherence is needed in this approach.

From~\cite{TL12}, it is known that $(B_{K_T,n},\|\cdot\|_{K_T,n})$ is a Banach space, we will let $ba_{K_T,n}$ be the topological dual of $B_{K_T,n}$, and let $\plus{ba_{K_T,n}}$ denote the positive linear forms, that is \[\plus{ba_{K_T,n}} := \lrcurly{\phi \in ba_{K_T,n}: \phi(X) \geq 0 \; \forall X \geq_{K_T} 0}.\]

\begin{defn}[Definition~4.1 of~\cite{TL12}]
\label{defn_stable}
A set $\Lambda \subseteq ba_{K_T,n}$ is called \textbf{\emph{$\Ft{t}$-stable}} if for all $\lambda \in L^\infty_t(\R_{+})$ and $\phi \in \Lambda$, the linear form $\phi^{\lambda}: X \ni B_{K_T,n} \mapsto \phi(\lambda X)$ is an element of $\Lambda$.
\end{defn}

\begin{thm}[Theorem 4.1 of~\cite{TL12}]
\label{thm_msb-dual}
Let $\seq{\tilde{R}}$ be a sequence of $\seq{\mathcal{S}^{d,n}}$-valued mappings on $B_{K_T,n}$.  Then the following are equivalent:
\begin{enumerate}
\item $\seq{\tilde{R}}$ is a conditionally coherent risk process.
\item There exists a nonempty $\sigma(ba_{K_T,n},B_{K_T,n})$-closed subset $\mathcal{Q}_t \neq \{0\}$ of $\plus{ba_{K_T,n}}$ which is $\Ft{t}$-stable and satisfies the equality
    \begin{equation}
    S_{\tilde{R}}^{\infty}(t,X) = \lrcurly{u \in M_t^n \cap B_{K_T,n}: \phi(X + u) \geq 0 \; \forall \phi \in \mathcal{Q}_t}.
    \end{equation}
\end{enumerate}
\end{thm}

We finish the discussion of the dual representation by considering the case when the risk process additionally satisfies a ``Fatou property'' as defined below.
\begin{defn}
\label{defn_msb-fatou}
A sequence $\seq{\tilde{R}}$ of $\seq{\mathcal{S}^{d,n}}$-valued mappings on $B_{K_T,n}$ is said to satisfy the \textbf{\emph{Fatou property}} if for all $X \in B_{K_T,n}$ and all times $t$
\[\limsup_{n \to +\infty} S_{\tilde{R}}^{\infty}(t,X_n) \subseteq S_{\tilde{R}}^{\infty}(t,X)\]
for any bounded sequence $(X_m)_{m \in \mathbb{N}} \subseteq B_{K_T,n}$ which converges to $X$ in probability.
\end{defn}
Note that in the above definition the limit superior is defined to be $\limsup_{n \to +\infty} B_n
= \cl\bigcup_{n \in \mathbb{N}} \bigcap_{m \geq n} B_m$ for a sequence of sets $(B_n)_{n \in \mathbb{N}}$.

For the following theorem we assume two additional properties on the convex cone $K_T$:
\renewcommand{\theenumi}{k\arabic{enumi}}
\begin{enumerate}
\setcounter{enumi}{3}
\item for almost every $\omega \in \Omega$: $\R^d_+ \backslash \{0\} \subseteq \interior[K_T(\omega)]$ or equivalently $\plus{K_T(\omega)} \backslash \{0\} \subseteq \interior[\R^d_+]$;
\item $K_T$ and $\plus{K_T}$ are both generated by a finite number of linearly independent and bounded generators denoted respectively by $(\xi_i)_{i = 1}^N$ and $(\plus{\xi_i})_{i = 1}^{N^+}$.
\end{enumerate}
\renewcommand{\theenumi}{\arabic{enumi}}

Let $L^{1,n}(\plus{K_T}) := \lrcurly{Z \in \LdpK{0}{}{\plus{K_T}}: \trans{1_{d,n}} Z \in \LoF{}}$.  In the following theorem we will use $L^{1,n}(\plus{K_T})$ as a dual space for $B_{K_T,n}$.  For $Z \in L^{1,n}(\plus{K_T})$, the linear form $\phi_Z(X) := \E{\trans{Z}X}$ belongs to $\plus{ba_{K_T,n}}$.  The norm $\|Z\|_{d,n} :=  \sup\lrcurly{|\E{\trans{Z}X}|: X \in B_{K_T,n}, \|X\|_{K_T,n} \leq 1}$ is the dual norm for any $Z \in L^{1,n}(\plus{K_T})$.

\begin{thm}[Theorem 4.2 of~\cite{TL12}]
\label{thm_msb-dual_fatou}
Let $\seq{\tilde{R}}$ be a conditionally coherent risk process on $B_{K_T,n}$ and let $K_T$ satisfy property $k1-k5$.  The following are equivalent:
\begin{enumerate}
\item For every time $t \in [0,T]$, there exists a closed conditional cone $\{0\} \neq \mathcal{Q}_t^1 \subseteq L^{1,n}(\plus{K_T})$ (in the norm topology, with norm $\|\cdot\|_{d,n}$) such that for any $X \in B_{K_T,n}$
    \begin{equation}
    S_{\tilde{R}}^{\infty}(t,X) = \lrcurly{u \in M_t^n \cap B_{K_T,n}: \forall Z \in \mathcal{Q}_t^1: \E{\trans{Z}(X + u)} \geq 0}.
    \end{equation}
\item $\seq{\tilde{R}}$ satisfies the Fatou property.
\item $C_t := \lrcurly{X \in B_{K_T,n}: 0 \in \tilde{R}_t(X)}$ is $\sigma(B_{K_T,n},L^{1,n}(\plus{K_T}))$-closed.
\end{enumerate}
\end{thm}

We conclude this section by discussing time consistency properties as they were defined in the measurable selector approach in \cite{TL12}.  As in the set-optimization  approach in the previous section one would like to define a property that is equivalent to a recursive form.  For this reason we will extend the risk process to be a function of a set.  For a set $\X \subseteq B_{K_T,n}$, let us define  $\tilde{R}_t(\X) \in \mathcal{S}_t^{d,n}$ via its selectors, that is
\[\LdpK{0}{t}{\tilde{R}_t(\X)} \cap B_{K_T,n} = \cl \env_{\Ft{t}} \bigcup_{X \in \X} S_{\tilde{R}}^{\infty}(t,X) =: S_{\tilde{R}}^{\infty}(t,\X),\]
where, for any $\Gamma \subseteq B_{K_T,n}$, $\env_{\Ft{t}} \Gamma$ denotes the smallest $\Ft{t}$-decomposable set (see
definition~\ref{defn_msb-acceptance}) which contains $\Gamma$.
This means that the measurable selectors of the risk process of a set are defined by the closed and $\Ft{t}$-decomposable version of the pointwise union.  Note that if $\X = \{X\}$ then this reduces to the prior definition on portfolios.
The risk process of a set is defined in this way because the selection risk measure must be closed and
$\Ft{t}$-decomposable-valued to guarantee the existence of an $\Ft{t}$-measurable random set $\tilde{R}_t(\X)$ such that $S_{\tilde{R}}^{\infty}(t,\X) = \LdpK{0}{t}{\tilde{R}_t(\X)} \cap B_{K_T,n}$.
\begin{defn}
\label{defn_tc}
A risk process $\seq{\tilde{R}}$ is called \textbf{\emph{consistent in time}} if for any $t,s \in [0,T]$ with $t < s$ and $X \in B_{K_T,n}$, ${\bf Y} \subseteq B_{K_T,n}$
\[\tilde{R}_{s}(X) \subseteq \tilde{R}_{s}({\bf Y}) \Rightarrow \tilde{R}_t(X) \subseteq \tilde{R}_t({\bf Y}).\]
\end{defn}

The following theorem gives equivalent characterizations of consistency in time, the last one being a recursion in the spirit of Bellman's principle.

\begin{thm}[Theorem 5.1 of~\cite{TL12}]
\label{thm_tc_equiv}
A normalized risk process $\seq{\tilde{R}}$ on $B_{K_T,n}$ is consistent in time if any of the following equivalent conditions hold:
\begin{enumerate}
\item If $\tilde{R}_{s}(X) \subseteq \tilde{R}_{s}({\bf Y})$ for $X \in B_{K_T,n}$ and ${\bf Y} \subseteq B_{K_T,n}$, then $\tilde{R}_t(X) \subseteq \tilde{R}_t({\bf Y})$ for $t \leq s \leq T$.
\item If $\tilde{R}_{s}(X) = \tilde{R}_{s}({\bf Y})$ for $X \in B_{K_T,n}$ and ${\bf Y} \subseteq B_{K_T,n}$, then $\tilde{R}_t(X) = \tilde{R}_t({\bf Y})$ for $t \leq s \leq T$.
\item For all $X \in B_{K_T,n}$, $S_{\tilde{R}}^{\infty}(t,X) = S_{\tilde{R}}^{\infty}(t,-S_{\tilde{R}}^{\infty}(s,X))$ for $t \leq s \leq T$.
\end{enumerate}
\end{thm}

\subsection{Set-valued portfolio approach}
\label{sec_fun-sets}

The approach for considering sets of portfolios, so called set-valued portfolios, as the argument of a set-valued risk measure was proposed in~\cite{CM13}. The reasoning for considering set-valued portfolios is to take the risk, not only of a portfolio $X$, but of
every possible portfolio that $X$ can be traded for in the market, into account.  We will denote by $\X$ the random set of
portfolios for which $X \in \LdpF{}$ can be exchanged.  The concept of set-valued portfolios appears naturally when trading opportunities in the market are taken into account.  Below we provide two examples, one in which no trading is allowed and another in which any possible trade can be used.  There are other examples provided in~\cite{CM13} on how a set-valued portfolio can be obtained, and the definition of the risk measure is independent of the method used to construct set-valued portfolios.
\begin{ex}
\label{ex_fun-sets_no-transaction}
The random mapping $\X = X + \R^d_-$ for a random vector $X \in \LdpF{}$ describes the case when no exchanges are allowed.
\end{ex}

\begin{ex}(Example 2.2 of~\cite{CM13})
The random mapping $\X = X + {\bf K}$ for a random vector $X \in \LdpF{}$ and a lower convex (random) set ${\bf K}$, such that $\LdpK{p}{}{{\bf K}}$ is closed, defines the set-valued portfolios related to the exchanges defined by ${\bf K}$.  If $K$ is a solvency cone (see e.g.~\cite{K99,S04,KS09}) or the sum of solvency cones at different time points, then ${\bf K} = -K$ is an exchange cone, and the associated random mapping defines a set-valued portfolio.
The setting of example~\ref{ex_fun-sets_no-transaction} corresponds to the
case where ${\bf K} = \R^d_-$.
\end{ex}

We will slightly adjust the definitions given in~\cite{CM13} to include the dynamic extension of such risk measures,  to incorporate  the set of eligible portfolios $M_t$, and go beyond the coherent case.

Let $\mathcal{S}_T^d$ denote the set of $\Ft{}$-random sets in $\R^d$ (as in section~\ref{sec_measurable} above). Let $\bar{\mathcal {S}}_T^d\subseteq \mathcal{S}_T^d$ be those random sets that are nonempty, closed, convex and lower, that is for $X\in\X$ also $Y\in\X$ whenever $X - Y \in \R^d_+$ $\P$-a.s. As in~\cite{CM13}, we will consider set-valued portfolios $\X \in \bar{\mathcal {S}}_T^d$. By proposition~2.1.5 and theorem 2.1.6 in \cite{M05}, the collection of $p$-integrable selectors of $\X$, that is $\LdpK{p}{}{\X}$, is a nonempty, closed, ($\Ft{}$-)conditionally convex, lower and $\Ft{}$-decomposable set,
which is an element of $\mathcal{G}(\LdpF{};\LdpF{-})$. In~\cite{CM13}, $\bar{\mathcal {S}}_T^d$ is used as the pre-image set, one could also use the family of sets of selectors $\{\LdpK{p}{}{\X}:\X \in \bar{\mathcal {S}}_T^d \}\subset \mathcal{G}(\LdpF{};\LdpF{-})$ as the pre-image set, which is particular useful when dynamic risk measures are considered and recursions due to multi-portfolio time consistency become important.
Recall that $\mathcal{P}(M_t;M_{t,+}) := \lrcurly{D \subseteq M_t: D = D + M_{t,+}}$ denotes the set of upper sets, which will be used as the image space for the risk measures. Closed (conditionally) convex risk measures map into $\mathcal{G}(M_t;M_{t,+})$.

In the following definition for convex risk measures we consider a modified version of set-addition used  in~\cite{CM13} which
is denoted by $\oplus$.  For two random sets $\X,{\bf Y} \in \mathcal{S}_T^d$, $\X \oplus {\bf Y} \in \mathcal{S}_T^d$ is the
random set defined by the closure of $\X[\omega] + {\bf Y}[\omega]$ for all $\omega \in \Omega$.
Note that, by proposition 2.1.4 in~\cite{M05}, if the probability space $(\Omega,\Ft{},\P)$ is non-atomic and $p \in [1,+\infty)$ then
$\LdpK{p}{}{\X \oplus {\bf Y}} = \cl\lrsquare{\LdpK{p}{}{\X} + \LdpK{p}{}{{\bf Y}}}$.

\begin{defn}[Definition 2.9 of~\cite{CM13}]
\label{defn_fun-sets_rm}
A function ${\bf R}_t:\bar{\mathcal {S}}_T^d\to \mathcal{P}(M_t;M_{t,+})$
is called a \textbf{\emph{set-valued conditional risk measure}} if it satisfies the following conditions.
\begin{enumerate}
\item Cash invariance: ${\bf R}_t(\X + m) = {\bf R}_t(\X) -
m$ for any $\X$ and $m \in M_t$.
\item Monotonicity: Let $\X \subseteq {\bf Y}$ almost surely,
then ${\bf R}_t({\bf Y}) \supseteq {\bf R}_t(\X)$.
\end{enumerate}

The risk measure ${\bf R}_t$ is said to be \textbf{\emph{closed-valued}} if its values are closed sets.

The risk measure ${\bf R}_t$ is said to be \textbf{\emph{(conditionally) convex}} if for every set-valued portfolio $\X,{\bf Y}$ and $\lambda \in [0,1]$
(respectively $\lambda \in \LiF{t}$ such that $0 \leq \lambda \leq 1$)
\[{\bf R}_t(\lambda \X \oplus (1-\lambda) {\bf Y}) \supseteq \lambda {\bf R}_t(\X) + (1-\lambda) {\bf R}_t({\bf Y}).\]

The risk measure ${\bf R}_t$ is said to be \textbf{\emph{(conditionally) positive homogeneous}} if for every $\X$ and $\lambda >
0$ (respectively $\lambda \in L^\infty_t(\R_{++})$)
\[{\bf R}_t(\lambda \X) = \lambda {\bf R}_t(\X).\]

The risk measure ${\bf R}_t$ is said to be \textbf{\emph{(conditionally) coherent}} if it is (conditionally) convex and (conditionally) positive homogeneous.
\end{defn}

The closed-valued variant of $R_t$ is denoted by $\bar{\bf R}_t(\X) = \cl ({\bf R}_t(\X))$ for every set-valued portfolio $\X\in\bar{\mathcal {S}}_T^d$.

A set-valued portfolio $\X$ is acceptable if $0 \in {\bf R}_t(\X)$, i.e. we can define the acceptance set ${\bf A}_t \subseteq\bar{\mathcal {S}}_T^d$ by ${\bf A}_t := \lrcurly{\X: 0 \in {\bf R}_t(\X)}$.  And a primal representation for the risk measures can be given by the usual definition ${\bf R}_t(\X) = \lrcurly{u \in M_t: \X + u \in {\bf A}_t}$ due to cash invariance.

We will now consider a subclass of set-valued conditional risk measures  presented in~\cite[section 3]{CM13} that are constructed using a scalar dynamic risk measure  for each component. For the remainder of this section we will consider the case when $M_t = \LdpF{t}$.  In~\cite{CM13}, only (scalar) law invariant coherent risk measures were considered for this approach, we will consider the more general case.

Let $\rho_t^1,...,\rho_t^d$ be dynamic risk measures defined on $\LpF{}$ with values in $L^p_t(\R \cup \{+\infty\})$.  For a random vector $X = \transp{X_1,...,X_d} \in \LdpF{}$ we define
\[{\boldsymbol\rho}_t(X) = \transp{\rho_t^1(X_1),...,\rho_t^d(X_d)}.\]
We say the vector $X \in \LdpF{}$ is \textbf{\emph{acceptable}} if ${\boldsymbol\rho}_t(X) \leq 0$, i.e. $\rho_t^i(X_i) \leq 0$ for all $i = 1,...,d$.  We say the set-valued portfolio $\X$ is \textbf{\emph{acceptable}} if there exists a $Z \in \LdpK{p}{}{\X}$ such that ${\boldsymbol\rho}_t(Z) \leq 0$.

\begin{defn}[Definition 3.3 of~\cite{CM13}]
\label{defn_constructive}
The \textbf{\emph{constructive conditional risk measure}} ${\bf R}_t: \bar{\mathcal {S}}_T^d\to \mathcal{P}(\LdpF{t};\LdpF{t,+})$ is defined for any set-valued portfolio $\X$ by \[{\bf R}_t(\X) = \lrcurly{u \in \LdpF{t}: \X + u \text{ is acceptable}},\]
which is equivalently to
\begin{equation}
\label{eq_constructive}
{\bf R}_t(\X) = \bigcup_{Z \in \LdpK{p}{}{\X}} ({\boldsymbol\rho}_t(Z) + \LdpF{t,+}).
\end{equation}
The closed-valued variant is defined by $\bar{\bf R}_t(\X) := \cl ({\bf R}_t(\X))$ for every $\X\in\bar{\mathcal {S}}_T^d$.
\end{defn}

In \cite{CM13}, the constructive (static) risk measures have been called selection risk measures, we modified the name here in accordance to the title of the paper  \cite{CM13} to avoid confusion with the measurable selector approach from section~\ref{sec_measurable}.

\begin{ex}
Consider the no-exchange set-valued portfolios from example~\ref{ex_fun-sets_no-transaction}.  Then the constructive conditional risk measure associated with any vector of scalar conditional risk measures is given by
\[{\bf R}_t(\X) = {\boldsymbol\rho}_t(X) + \LdpF{t,+}.\]
\end{ex}

\begin{thm}[Theorem 3.4 of~\cite{CM13}]
\label{thm_constructive}
Let ${\boldsymbol\rho}_t$ be a vector of dynamic risk measures, then ${\bf R}_t$ and $\bar{\bf R}_t$ given in definition~\ref{defn_constructive}   are both set-valued conditional risk measures.

If ${\boldsymbol\rho}_t$ is convex (conditionally convex, positive homogeneous, conditionally positive homogeneous, law invariant convex on an non-atomic probability space), then ${\bf R}_t$ and $\bar{\bf R}_t$ are convex (conditionally convex, positive homogeneous, conditionally positive homogeneous, law invariant convex on an non-atomic probability space).
\end{thm}

Furthermore, \cite{CM13} gives conditions under which the constructive (static) risk measure ${\bf R}_0$ defined in \eqref{eq_constructive} in the coherent case is closed, or Lipschitz and deduces upper and lower bounds for it and dual representations in certain special cases. Numerical examples for the calculation of upper and lower bounds are given.

\subsection{Family of scalar risk measures}
\label{sec_scalar}

Consider $\LdpF{}$, $p \in [1,+\infty]$ with the norm topology for $p \in [1,+\infty)$ and the weak* topology for $p = +\infty$. Recall from definition~\ref{defn_acceptance} that
a set $A_t \subseteq \LdpF{}$ is a \textbf{\emph{conditional acceptance set}} at time $t$ if it satisfies
$M_t \cap A_t \neq \emptyset$, $\tilde{M}_t[\omega] \cap (\R^d \backslash A_t[\omega]) \neq \emptyset$ for almost every $\omega \in \Omega$, and $A_t + \LdpF{+} \subseteq A_t$.

We will define a family of scalar conditional risk measures $\rho_t^{w}$ with parameter $w \in \plus{M_{t,+}} \backslash \prp{M_t}$ via their primal representation.  The scalar risk measures map into the random variables with values in the extended real line, that is, into the space $L^0_t(\bar\R)$ with $\bar\R:=\R \cup \{\pm\infty\}$.

\begin{defn}
\label{defn_multi-scalar}
A function $\rho_t^{w}: \LdpF{} \to L^0_t(\bar\R)$ satisfying
\begin{align}
\label{scalarRM}
\rho_t^{w}(X) = \essinf\lrcurly{\trans{w}u: u \in M_t, X + u \in A_t}
\end{align}
for a parameter $w \in \plus{M_{t,+}} \backslash \prp{M_t}$ and a conditional acceptance set $A_t$ is called a \textbf{\emph{multiple asset conditional risk measure}} at time $t$.
\end{defn}

Clearly, the scalar risk measures defined above  are scalarizations of a set-valued risk measure from the set-optimization  approach (see section~\ref{sec_vector}) defined by $R_t:=\{u \in M_t: X + u \in A_t\}$, where the scalarizations are taken with respect to vectors $w \in \plus{M_{t,+}} \backslash \prp{M_t}$, that is
\begin{align}
\label{scalarization}
\rho_t^{w}(X) = \essinf_{u \in R_t(X)} \trans{w}u=
   \essinf \lrcurly{\trans{w}u: u \in M_t, X + u \in A_t}.
\end{align}

Note, that when $R_t$ is $K$-compatible (that is $A_t=A_t+\LdpK{p}{t}{K}$) for some $\Ft{t}$-measurable random cone $K\subseteq \tilde{M}_t$, then $\rho_t^{w}(X)[\omega] = -\infty$ on $w(\omega) \not\in \plus{K[\omega]}$ for any $X \in \LdpF{}$.  Thus, one can restrict oneself in this case to parameters $w$ in the basis of $\LdpK{q}{t}{\plus{K}} \backslash \prp{M_t}$.

We will give some examples from the literature of scalar risk measures of form \eqref{scalarRM}.
\begin{ex}
In \cite{FMM13,ADM09,FS06,K09,Sc04} risk measures of form \eqref{scalarRM} have been studied in the static case.

In a frictionless market let the time $t$ prices be given by the (random) vector $S_t$. In this case the solvency cones
(see~\cite{K99,S04,KS09}) $\seq{K}$ are given by $K_t[\omega] = \lrcurly{x \in \R^d: \trans{S_t(\omega)}x \geq 0}$, where the normal vector $S_t(\omega)$ is the unique vector in the basis of $\plus{K_t[\omega]}$.  Let $A_t = A_t +\LdpK{p}{t}{K_t \cap \tilde{M_t}} + \LdpK{p}{}{K_T}$,
\[
\rho_t^{S_t}(X) = \essinf\lrcurly{\trans{S_t}u: u \in M_t, X + u \in A_t} =
\tilde{\rho}_t^{S_t}(\trans{S_T}X)
\]
for any $X \in \LdpF{}$ (since $\LdpK{q}{t}{\plusp{K_t\cap\tilde{M}_t}} := \LdpK{q}{t}{\plus{K_t}} + \prp{M_t}$). It can be seen that $\tilde{\rho}_t^{S_t}(Z) = \essinf\lrcurly{\trans{S_t}u: u \in M_t, Z + \trans{S_T}u \in
\tilde{A}_t}$ with $\tilde{A}_t = \lrcurly{\trans{S_T}X: X \in A_t}$ is the dynamic version of
the risk measures with multiple eligible assets defined in~\cite{FMM13,ADM09,FS06,K09,Sc04} (and with single eligible assets (which is not necessarily the original num\'{e}raire) defined in~\cite{FMM12a,FMM12b}).
$\tilde{A}_t$ satisfies definition~1 of~\cite{FMM13} for an acceptance set.
\end{ex}

\begin{ex}
\cite{BR06} discusses scalar static risk measure of multivariate claims, when only a single eligible asset is considered, that is
\[\rho(X) = \inf\lrcurly{m \in \R: X + m e_1 \in A}\]
for $X \in \LdiF{}$, where $A \subseteq \LdiF{}$ is an acceptance set.  We can see that this has the form $\rho(X)  = \inf\lrcurly{\trans{e_1}u: u \in \R \times \{0\}^{d-1}, X + u \in A}$, i.e. the scalarization of a set-valued risk measure with $M_0 = \R \times \{0\}^{d-1}$ and $w = e_1$.
\end{ex}

\begin{ex}
In \cite{WA13} so called liquidity-adjusted risk measure  $\rho^V: \LdiF{} \to \R$, which are scalar static risk measure of multivariate claims in markets with frictions, are studied, when only a single eligible asset is considered.
The primal representation
\[ \rho^V(X) = \inf \{  k\in \R: X + k e_1\in A^V \} \]
for
$A^V := \{  X\in \LdiF{} : V(X) \in A \}  $, where $V$ is a real valued function providing the  value of a portfolio $X$ under liquidity and portfolio constraints and $A\subseteq \LiF{}$ is the acceptance set of a scalar convex risk measure in the sense of \cite{FS02}. Clearly, $\rho^V(X)$ is of form \eqref{scalarRM}.
\end{ex}

\begin{ex}
\label{exSHP}
In \cite{BLPS92,BV92,PL97,JK95,R08} (and many other papers) the scalar superhedging price in a market with two assets and transaction costs has been studied. The $d$ asset case is treated in \cite{RZ11,LR11}. Let $\seq{K}$ be the sequence of solvency cones modeling the market with proportional transaction costs.

The $d$ dimensional version of the dual representation of the scalar superhedging price given in Jouini, Kallal \cite{JK95} reads as follows. Let $X\in \LdpF{}$ be a payoff in physical units.
Under an appropriate robust no arbitrage condition, the scalar superhedging price $\pi^a_i(X)$ in units of asset $i\in\{1,...,d\}$ at time $t=0$ is given by
\begin{align}
    \label{JK d assets}
   \pi^a_i(X)= \sup_{(S_t, \Q)\in\mathcal Q^i} \EQ{\trans{S_T}X},
\end{align}
where $\mathcal Q^i$ is the set of all processes $(S_t)_{t=0}^T$ and their equivalent martingale measures $\Q$ with $\frac{d\Q}{d\P}\in L^1(\mathcal F_T)$,
$S_t^i\equiv 1$, $\Et{\frac{d\Q}{d\P}}{t} S_t\in \Ldq{\Ft{t};\plus{K_t}}$  for all $t$.
Theorem~6.1 in \cite{LR11} shows that  \eqref{JK d assets} can be obtained by scalarizing the coherent set-valued risk measure with acceptance set $A_0=\sum_{s=0}^T \LdpK{p}{s}{K_s}$ and single eligible asset (asset $i$, which is also the num\'{e}raire asset, i.e. $M_0 = \lrcurly{m \in \R^d: m_j = 0 \; \forall j \neq i}$) w.r.t. the unit vector $w=e_i \in (K_0\cap M_0)^+$. Thus, $\pi^a_i$ is a special case of \eqref{scalarRM}.
\end{ex}

Of course any standard scalar risk measure in a frictionless markets with single eligible asset as in \cite{FS02,AD99} is also special cases of \eqref{scalarRM}, but in that case there is no advantage to explore the relationship with a set-valued risk measure via \eqref{scalarization}. In any other case, i.e. if one of the following is considered: multiple eligible assets, multivariate claims, transaction costs or other market frictions, it can be advantageous to explore \eqref{scalarization} as the dual representation of the corresponding set-valued risk measure given in section~\ref{sec_vector} can lead to a dual representation of the scalarization as demonstrated in \eqref{JK d assets}. Furthermore, even if one is interested in only one particular scalarization (as it is the case in all the examples above), the dual representation of the scalar risk measure might involve the whole family of scalarizations (as in example~\ref{exSHP}, where the constraints $S_t\in \plus{K_t}$ a.s. for all $t$ enter the scalar problem in \eqref{JK d assets}). This is related to time consistency properties of the scalar risk measure and multi-portfolio time consistency of the corresponding set-valued risk measure (see definition~\ref{defn_mptc}). In this paper we are only concerned with the connection between a family of scalar risk measures and a set-valued risk measure. Lemma~\ref{lemma_scalar_rep}
below gives very mild conditions under which a set-valued risk measures can be equivalently represented by a family of scalar risk
measures. Results about dual representations and the study of time consistency properties of the family of scalar risk measures are left for further research.

The main motivation to study a family of scalar risk measures in this section is that it allows to generalize all of the examples given above in a unified way by allowing multiple eligible assets, multivariate claims and frictions in the form of transaction costs, as well as considering a dynamic setting. As example~\ref{exSHP} suggests, viewing a scalar risk measure in a market with frictions as being a scalarization of a set-valued risk measure has the advantage of obtaining dual representations and conditions on time consistency by using the corresponding results of the set-valued risk measure.

A different approach concerning a family of scalar risk measures and multiple eligible assets in a frictionless market was taken in~\cite{JB08}.  In that paper, given a set of
eligible assets (with values $S_T^i$ for $i = 1,...,n$), the risk of the portfolio $X$ is the set of values $\lrcurly{\sum_{i = 1}^n
\rho_t^{S_T^i}(X_i) S_T^i: X = \sum_{i = 1}^n X_i}$ where $\rho_t^{S_T^i}$ is a risk measure in asset $i$ (with change of
num\'{e}raire).  However, we will not discuss this approach further since lemma 4.10 of that paper demonstrates that
$\rho_t^{S_T^0}(X) \leq \rho_t^{S_T^0}(-\sum_{i = 1}^n \rho_t^{S_T^i}(X_i) S_T^i)$ for any choice of num\'{e}raire $0$ and any allocation of $X = \sum_{i = 1}^n X_i$,
i.e. the family of risks (as a portfolio) has risk bounded below by the risk of the initial portfolio no matter the num\'{e}raire
chosen.

In the following proposition we show that the multiple asset conditional scalar risk measures satisfy monotonicity and a translative
property.  These properties are usually given as the definition of a risk measure in the literature given in the above examples. However, here we consider the primal representation~\ref{defn_multi-scalar} as the starting point.

\begin{prop}
\label{prop_multi-scalar_properties}
Let $\rho_t^{w}: \LdpF{} \to L^0_t(\bar\R)$ be a multiple asset conditional scalar risk measure at time $t$ for pricing vector $w \in
\plus{M_{t,+}} \backslash \prp{M_t}$.  Then $\rho_t^{w}$ satisfies the following conditions.
\begin{enumerate}
\item If $Y - X \in \LdpF{+}$ for $X,Y \in \LdpF{}$, then $\rho_t^{w}(Y) \leq \rho_t^{w}(X)$.
\item $\rho_t^{w}(X + m) = \rho_t^{w}(X) - \trans{w}m$ for all $X \in \LdpF{}$ and $m \in M_t$.
\end{enumerate}

Further, if we consider the family of such risk measures over all pricing vectors $w \in \plus{M_{t,+}} \backslash \prp{M_t}$ then we
have the following finiteness properties.
\begin{enumerate}
\setcounter{enumi}{2}
\item $\rho_t^{w}(0) < +\infty$ for every $w \in \plus{M_{t,+}} \backslash \prp{M_t}$.
\item $\rho_t^{w}(0) > -\infty$ for some $w \in \plus{M_{t,+}} \backslash \prp{M_t}$.
\end{enumerate}
\end{prop}
\begin{proof}
Let $\rho_t^{w}(X) := \essinf \lrcurly{\trans{w}u: u \in M_t, X + u \in A_t}$ for every $X \in \LdpF{}$, every $w \in
\plus{M_{t,+}} \backslash \prp{M_t}$, and some conditional acceptance set $A_t$.
\begin{enumerate}
\item Let $X,Y \in \LdpF{}$ such that $Y - X \in \LdpF{+}$.  Let $w \in \plus{M_{t,+}} \backslash \prp{M_t}$.
    \begin{align*}
    \rho_t^{w}(Y) &= \essinf \lrcurly{\trans{w}u: u \in M_t, Y + u \in A_t} = \essinf \lrcurly{\trans{w}u: u \in M_t, X + (Y-X) + u
\in A_t}\\
    &\leq \essinf \lrcurly{\trans{w}u: u \in M_t, X + u \in A_t} = \rho_t^{w}(X).
    \end{align*}
\item Let $X \in \LdpF{}$ and $m \in M_t$.  Let $w \in \plus{M_{t,+}} \backslash \prp{M_t}$.
    \begin{align*}
    \rho_t^{w}(X + m) &= \essinf \lrcurly{\trans{w}u: u \in M_t: X + m + u \in A_t}\\
    &= \essinf \lrcurly{\trans{w}(u - m): u \in M_t, X + u \in A_t} = \rho_t^{w}(X) - \trans{w}m.
    \end{align*}
\item Fix some $\omega \in \Omega$.  $\rho_t^{w}(0)[\omega] = +\infty$ for some $w \in \plus{M_{t,+}} \backslash \prp{M_t}$ if and only if $A_t[\omega] \cap \tilde{M}_t[\omega] = \emptyset$, which by $A_t \cap M_t \neq \emptyset$ is false.
\item Fix some $\omega \in \Omega$.  $\rho_t^{w}(0)[\omega] = -\infty$ for every $w \in \plus{M_{t,+}} \backslash \prp{M_t}$ if and only if $(\R^d \backslash A_t[\omega]) \cap \tilde{M}_t[\omega] = \emptyset$, which by definition is false.
\end{enumerate}
\end{proof}

\section{Relation between approaches}
\label{sec_relation}

In this section we compare the properties for each of the techniques for dynamic multivariate risk measures.  It will be shown that the set-valued portfolio approach to dynamic risk measures is the most general model into which every other approach can be embedded.  It will be shown in section~\ref{sec_vector-portfolio} that under weak assumptions on the construction of the set-valued portfolios, the set-optimization approach is equivalent to the set-valued portfolio approach.  Because additional properties for dynamic risk measures have been studied previously for the set-optimization approach and due to the (often) one-to-one relation with the set-portfolio approach, we will present the relations in this section as comparisons with the \emph{set-optimization approach}.

\subsection{Set-optimization approach versus measurable selectors}
\label{sec_vector-measurable}

In order to compare these two approaches, one first needs to agree on the same preimage and image space. One possibility would be to define the risk measures of section~\ref{sec_vector} on the space $B_{K_T,n}$. This can be done as the theory involved (set-optimization) works for any locally convex space as the preimage space. The other possibility is to consider the measurable selectors approach of section~\ref{sec_measurable} on $\LdpF{}$ spaces. This in not a problem for the definition of risk processes given in definition~\ref{defn_risk_process}, but could pose a problem for primal and dual representations,
see discussion in remark~\ref{rem_preimage} for more details. However, since for the comparison results we just work with the definitions, we will follow this path here. Thus, consider $\LdpF{}$ spaces for $p \in [0,+\infty]$ endowed with the metric topology (that is the norm topology for $p\geq 1$), even for $p=+\infty$ which is in contrast to \cite{FR12,FR12b} where the weak* topology is used for $p=+\infty$.
Also, as the definition of the risk process does not rely on the space of eligible portfolios to be $M_t^n$, we will use a general space of eligible portfolios $M_t$. We will show that when the dynamic risk measure has closed and conditionally convex images, the set-optimization  and the measurable selectors approach coincide.

\begin{rem}
\label{rem_preimage}
While the space $B_{K_T,n}$ shares many properties with $\LdiF{}$, the two do not coincide in general.  If $n = d$ or additional
assumptions (e.g. substitutability from~\cite{JMT04}) are satisfied, then $\LdiF{} \subseteq B_{K_T,n}$. If  $n = d$ and $K_T = \R^d_+$ almost surely, then $B_{K_T,n}=\LdiF{}$. However, in general the two spaces are not comparable in the set-inclusion relation.
Therefore, without additional assumptions, it is not trivial to use the representation results from~\cite{TL12} for the space $\LdiF{}$. Furthermore, the assumptions for the Fatou duality (theorem~\ref{thm_msb-dual_fatou}) exclude the special case $K_T = \R^d_+$ and thus exclude the case $B_{K_T,n}=\LdiF{}$ when $n = d$.
However, the definition for risk process can be given for $\LdpF{}$ spaces (and this is used in this section). But complications arise in both, the
primal and dual definition, as e.g. boundedness is used in the proofs in \cite{TL12}. \end{rem}

The definition for $\Ft{t}$-decomposability given below can be found in \cite[page 148]{M05} or \cite[page 260]{KS09}.
\begin{defn}
\label{defn_decomposable}
A set $D \subseteq \LdpF{}$ is said to be \textbf{\emph{$\Ft{t}$-decomposable}} if for any finite partition $(\Omega_t^n)_{n = 1}^N
\subseteq \Ft{t}$ of $\Omega$ and any family $(X_n)_{n = 1}^N\subseteq D$ for $N \in \mathbb{N}$, we have $\sum_{n = 1}^N 1_{\Omega_t^n} X_n \in
D$.
\end{defn}

The following theorem and corollary~\ref{cor_decomposable} below state that there is a one-to-one relation between conditional risk measures $R_t$
with closed and $\Ft{t}$-decomposable images and closed risk processes $\tilde{R}_t$.
In corollary~\ref{cor_convex_decomposable} we demonstrate that any conditional risk measure with closed and conditionally convex images also has $\Ft{t}$-decomposable images.

For notational purposes, let $\mathcal{S}_t :=\{\Gamma \in\mathcal{S}_t^d:\Gamma(\omega) \subseteq \tilde{M}_t[\omega]\;\as\}\subseteq \mathcal{S}_t^d$.
\begin{thm}
\label{thm_decomposable}
Let $\tilde{R}_t: \LdpF{} \to \mathcal{S}_t$ be a risk process at time $t$ (see definition~\ref{defn_risk_process}), then $R_t:
\LdpF{} \to \mathcal{P}(M_t;M_{t,+})$, defined by $R_t(X) := \LdpK{p}{t}{\tilde{R}_t(X)}$ for any $X \in \LdpF{}$, is a
conditional risk measure at time $t$ (see definition~\ref{defn_conditional}) with $\Ft{t}$-decomposable images.

Let $R_t: \LdpF{} \to \mathcal{P}(M_t;M_{t,+})$ be a conditional risk measure at time $t$ (see definition~\ref{defn_conditional})
with closed and $\Ft{t}$-decomposable images, then there exists a risk process $\tilde{R}_t: \LdpF{} \to \mathcal{S}_t$ (see definition~\ref{defn_risk_process}) such that $R_t(X) = \LdpK{p}{t}{\tilde{R}_t(X)}$ for any $X \in \LdpF{}$.
\end{thm}
\begin{proof}
\begin{enumerate}
\item Let $\tilde{R}_t: \LdpF{} \to \mathcal{S}_t$ be a risk process at time $t$.  Let $R_t: \LdpF{} \to \mathcal{P}(M_t;M_{t,+})$ be defined by $R_t(X) := \LdpK{p}{t}{\tilde{R}_t(X)}$ for any $X \in \LdpF{}$.  It remains to show that $R_t$ is a conditional risk measure at time $t$.
$\LdpF{+}$-monotonicity: let $X,Y \in \LdpF{}$ such that $Y - X \in \LdpF{+}$, then $\tilde{R}_t(Y) \supseteq \tilde{R}_t(X)$, and thus $R_t(X) \supseteq R_t(X)$.
$M_t$-translativity: let $X \in \LdpF{}$ and $m \in M_t$, then $R_t(X+m) = \LdpK{p}{t}{\tilde{R}_t(X+m)} = \LdpK{p}{t}{\tilde{R}_t(X)-m} = \LdpK{p}{t}{\tilde{R}_t(X)} - m = R_t(X) - m$.
Finiteness at zero: By $\tilde{R}_t(0) \neq \emptyset$ almost surely then trivially $R_t(0) = \LdpK{p}{t}{\tilde{R}_t(0)} \neq \emptyset$.  By $\tilde{R}_t(0) \neq \tilde{M}_t$ almost surely then if $u(\omega) \in \tilde{M}_t[\omega]\backslash\tilde{R}_t(0)[\omega]$ for almost every $\omega \in \Omega$ such that $u \in M_t$, then $u(\omega) \not\in R_t(0)[\omega]$ for almost every $\omega \in \Omega$.
$\Ft{t}$-decomposable images: Let $(\Omega_t^n)_{n = 1}^N \subseteq \Ft{t}$ for some $N \in \mathbb{N}$ be a finite partition of
$\Omega$ and let $(u_n)_{n = 1}^N \subseteq R_t(X)$ then $\sum_{n = 1}^N 1_{\Omega_t^n} u_n \in M_t$, then since $R_t(X)$ are the
measurable selectors of $\tilde{R}_t(X)$ it immediately follows that $\sum_{n = 1}^N 1_{\Omega_t^n} u_n \in R_t(X)$.
\item Let $R_t: \LdpF{} \to \mathcal{P}(M_t;M_{t,+})$ be a conditional risk measure at time $t$ with closed and $\Ft{t}$-decomposable
images.  By proposition 5.4.3 in~\cite{KS09} (for $p \in [0,+\infty)$) and theorem 1.6 of chapter 2 from~\cite{M05} (for $p =
+\infty$),
it follows that $R_t(X) = \LdpK{p}{t}{\tilde{R}_t(X)}$ for some almost surely closed random set $\tilde{R}_t(X)$ for every $X \in \LdpF{}$.  Trivially, we can see that $\tilde{R}_t(X) \subseteq \tilde{M}_t$ almost surely. It remains to show that $\tilde{R}_t$ is a risk process at time $t$.
 Let $X \in \LdpF{}$,
        then $\tilde{R}_t(X)$ is a closed $\Ft{t}$-measurable random set \cite[proposition 5.4.3]{KS09} and \cite[chapter 2 theorem 1.6]{M05}.
 Finiteness at zero of $R_t$ implies finiteness at zero of $\tilde{R}_t$.
Consider $X,Y \in \LdpF{}$ with $Y - X  \in \LdpF{+}$, then $R_t(Y) \supseteq R_t(X)$, which implies that  $\tilde{R}_t(Y) \supseteq \tilde{R}_t(X)$.
Let $X \in \LdpF{}$ and $m \in M_t$, then $R_t(X + m) = R_t(X) - m$. This implies $\LdpK{p}{t}{\tilde{R}_t(X + m)} = \LdpK{p}{t}{\tilde{R}_t(X)} - m = \LdpK{p}{t}{\tilde{R}_t(X)-m}$, i.e. $\tilde{R}_t(X + m) = \tilde{R}_t(X)-m$ almost surely.

\end{enumerate}
\end{proof}

In the below corollaries the conditional risk measure associated with the risk process (and vice versa) is defined as in theorem~\ref{thm_decomposable} above.
\begin{cor}
\label{cor_decomposable}
Let $\tilde{R}_t: \LdpF{} \to \mathcal{S}_t$ be a conditionally convex (conditionally positive homogeneous, normalized) risk process at time $t$ then the associated conditional risk measure is conditionally convex (conditionally positive homogeneous, normalized).

Let $R_t: \LdpF{} \to \mathcal{P}(M_t;M_{t,+})$ be a conditionally convex (conditionally positive homogeneous, normalized) conditional risk
measure at time $t$ with closed and $\Ft{t}$-decomposable images, then the associated risk process is conditionally convex (conditionally positive homogeneous, normalized).
\end{cor}
\begin{proof}
\begin{enumerate}
\item Let $\tilde{R}_t: \LdpF{} \to \mathcal{S}_t$ be a risk process at time $t$ and $R_t$ be the associated conditional risk measure.
Let $\tilde{R}_t$ be conditionally convex.  Take $X,Y \in \LdpF{}$, $\lambda \in \LiF{t}$ with $0 \leq \lambda \leq 1$. Then,
        \begin{align*}
        \lambda R_t(X) + (1-\lambda)R_t(Y)
        & = \lambda \LdpK{p}{t}{\tilde{R}_t(X)} + (1-\lambda) \LdpK{p}{t}{\tilde{R}_t(Y)}\\
        & = \LdpK{p}{t}{\lambda \tilde{R}_t(X) + (1-\lambda) \tilde{R}_t(Y)}\\
        &\subseteq \LdpK{p}{t}{\tilde{R}_t(\lambda X + (1-\lambda) Y)} = R_t(\lambda X + (1-\lambda) Y).
        \end{align*}
    Let $\tilde{R}_t$ be conditionally positive homogeneous.  Take $X \in \LdpF{}$ and $\lambda \in L^\infty_t(\R_{++})$. Then,
        $
        \lambda R_t(X) = \lambda \LdpK{p}{t}{\tilde{R}_t(X)} = \LdpK{p}{t}{\lambda \tilde{R}_t(X)} = \LdpK{p}{t}{\tilde{R}_t(\lambda X)} = R_t(\lambda X).
       $
   Let $\tilde{R}_t$ be normalized and let $X \in \LdpF{}$. Then,
       $
        R_t(X) + R_t(0) = \LdpK{p}{t}{\tilde{R}_t(X)} + \LdpK{p}{t}{\tilde{R}_t(0)}
        = \LdpK{p}{t}{\tilde{R}_t(X) + \tilde{R}_t(0)} = \LdpK{p}{t}{\tilde{R}_t(X)} = R_t(X).
       $
\item Let $R_t: \LdpF{} \to \mathcal{P}(M_t;M_{t,+})$ be a conditional risk measure at time $t$ and let $\tilde{R}_t$ be the associated risk process.
    Let $R_t$ be conditionally convex. Take $X,Y \in \LdpF{}$ and $\lambda \in \LiF{t}$ with $0 \leq \lambda \leq 1$. Then,
        \begin{align*}
        \LdpK{p}{t}{\lambda \tilde{R}_t(X) + (1-\lambda) \tilde{R}_t(Y)} &= \lambda R_t(X) + (1-\lambda)R_t(Y)\\
        &\subseteq R_t(\lambda X + (1-\lambda) Y) = \LdpK{p}{t}{\tilde{R}_t(\lambda X + (1-\lambda) Y)}.
        \end{align*}
         By \cite[chapter 2 proposition 1.2 (iii)]{M05} it holds $\lambda \tilde{R}_t(X) + (1-\lambda) \tilde{R}_t(Y) \subseteq \tilde{R}_t(\lambda X + (1-\lambda) Y)$ almost surely.
   The proof for conditional positive homogeneity and normalization is analog.
   \end{enumerate}
\end{proof}

As discussed in sections~\ref{sec_vector} and~\ref{sec_measurable}, we have time consistency properties for both the set-optimization  and measurable selector approach to risk measures.  Therefore, we would like to be able to compare multi-portfolio time consistency (definition~\ref{defn_mptc}) and consistency in time (definition~\ref{defn_tc}).  These properties coincide in their notation, however as we will show below the two properties only coincide under additional assumptions.
\begin{cor}
\label{cor_tc}
Let $\seq{\tilde{R}}$ be a normalized conditionally convex consistent in time risk process, then the associated dynamic risk measure is multi-portfolio time consistent if it is convex upper continuous.

Let $\seq{R}$ be a normalized multi-portfolio time consistent dynamic risk measure with closed and $\Ft{t}$-decomposable images for
all times $t$, then the associated risk process is consistent in time.
\end{cor}
\begin{proof}
\begin{enumerate}
\item Let $\seq{\tilde{R}}$ be a normalized conditionally convex risk process which is consistent in time such that the associated dynamic risk measure $\seq{R}$ is convex upper continuous.  By
    theorem~\ref{thm_tc_equiv}, it follows that $R_t(X) = \cl \env_{\Ft{t}} \bigcup_{Z \in R_{s}(X)} R_t(-Z)$ for any $X \in \LdpF{}$
    and any times $t,s \in [0,T]$ such that $t \leq s$.  By corollary~\ref{cor_decomposable} above, $\seq{R}$ is conditionally convex.

    We will show that the recursive form $\bigcup_{Z \in R_{s}(X)} R_t(-Z)$ is $\Ft{t}$-decomposable.  Let $N \in \mathbb{N}$,
$(u_n)_{n =1}^N \subseteq \bigcup_{Z \in R_{s}(X)} R_t(-Z)$ and $(\Omega_t^n)_{n=1}^N \subseteq \Ft{t}$ is a partition of $\Omega$.
Denote by $Z_n \in R_{s}(X)$ the element such that $u_n \in R_t(-Z_n)$ for every $n \in \{1,...,N\}$.  By
lemma~\ref{lemma_partition}, it follows that $\sum_{m =1}^N 1_{\Omega_t^m} Z_m \in R_{s}(X)$.  Then we can see
    \begin{align*}
    \sum_{n =1}^N 1_{\Omega_t^n}u_n &\in \sum_{n =1}^N 1_{\Omega_t^n} R_t(Z_n) = \sum_{n =1}^N 1_{\Omega_t^n} R_t(1_{\Omega_t^n} Z_n)\\
    &= \sum_{n =1}^N 1_{\Omega_t^n} R_t(1_{\Omega_t^n} \sum_{m =1}^N 1_{\Omega_t^m} Z_m) = \sum_{n =1}^N 1_{\Omega_t^n} R_t(\sum_{m
=1}^N 1_{\Omega_t^m} Z_m)\\
    &\subseteq \lrcurly{u \in M_t: \exists J \subseteq \{1,...,N\}: \P(\cup_{j \in J} A_j) = 1, \forall j \in J: 1_{\Omega_t^j} u \in
1_{\Omega_t^j} R_t(\sum_{m =1}^N 1_{\Omega_t^m} Z_m)}\\
    &= R_t(\sum_{m =1}^N 1_{\Omega_t^m} Z_m) \subseteq \bigcup_{Z \in R_{s}(X)} R_t(-Z).
    \end{align*}
    In the above we use the local property for conditionally convex risk measures (\cite[proposition 2.8]{FR12}) and lemma~\ref{lemma_partition}.
    Therefore, $\bigcup_{Z \in R_{s}(X)} R_t(-Z)$ is $\Ft{t}$-decomposable, and
    thus $R_t(X) = \cl \bigcup_{Z \in R_{s}(X)} R_t(-Z)$.  And as seen in~\cite[appendix B]{FR12b}, if $\seq{R}$ is convex upper continuous then $\bigcup_{Z \in R_{s}(X)} R_t(-Z)$ is closed for any $X \in \LdpF{}$. Therefore, $R_t(X) = \bigcup_{Z \in R_{s}(X)} R_t(-Z)$, i.e. $R_t(X)$ multi-portfolio time consistent.
\item Let $\seq{R}$ be a normalized multi-portfolio time consistent dynamic risk measure with closed and $\Ft{t}$-decomposable images
    for all time $t$.  Let $\seq{\tilde{R}}$ be the associated risk process.  By theorem~\ref{thm_mptc_equiv}, it follows that $R_t(X) =
    \bigcup_{Z \in R_{s}(X)} R_t(-Z)$ for any $X \in \LdpF{}$ and any times $t,s \in [0,T]$ such that $t \leq s$.  Since
    $R_t$ has closed and $\Ft{t}$-decomposable images then it additionally follows that $\bigcup_{Z \in R_{s}(X)} R_t(-Z) = \cl \env_{\Ft{t}}
    \bigcup_{Z \in R_{s}(X)} R_t(-Z)$ for any $X \in \LdpF{}$.  Therefore, $\LdpK{p}{t}{\tilde{R}_t(X)} =
    \LdpK{p}{t}{\tilde{R}_t(-R_{s}(X))}$ and thus, by theorem~\ref{thm_tc_equiv}, it follows that $\seq{\tilde{R}}$ is consistent in
    time.
\end{enumerate}
\end{proof}

The convex upper continuity in the first part of the above theorem could we weakened as one only needs $\bigcup_{Z \in R_{s}(X)} R_t(-Z)$ is closed for any $X \in \LdpF{}$ and $t\leq s$.

Up to this point we have made the additional assumption for conditional risk measures of section~\ref{sec_vector} to be
$\Ft{t}$-decomposable.  The following results (lemma~\ref{lemma_partition} and corollary~\ref{cor_convex_decomposable} below) demonstrate that a conditional risk measure with closed and conditionally convex images
satisfies a property stronger than $\Ft{t}$-decomposable images as the property remains true for any (possibly uncountable) partition as well.

\begin{lemma}
\label{lemma_partition}
Let $\seq{R}$ be a dynamic risk measure with closed and conditionally convex images.  Let $(A_i)_{i \in I} \subseteq \Ft{t}$ be a partition of $\Omega$.  Then
\[R_t(X) = \lrcurly{u \in M_t: \exists J \subseteq I \mbox{ with }\P(\cup_{j \in J} A_j) = 1\mbox{ such that } 1_{A_j} u \in 1_{A_j} R_t(X)\; \forall j \in J} \]
for any $X \in \LdpF{}$ and any time $t$.
\end{lemma}

Before giving the proof we give a remark on the uncountable summation as it will be used in part 2 (b) of the proof.
\begin{rem}\label{rem_summ}
As given in~\cite[Chapter~3~Section~5]{B66} and~\cite[Chapter~3~Section~3.9]{C66}, the arbitrary summation on a Hausdorff commutative topological group is given by $\sum_{j \in J} f_j = \lim_{K \in \mathcal{J}} \sum_{k \in K} f_k$, for any $\{f_j \in \mathcal{X}: j \in J\}$ where $\mathcal{X}$ is a Hausdorff commutative topological group, such that $\mathcal{J} = \lrcurly{K \subseteq J: \#K < +\infty}$, i.e. $\mathcal{J}$ are the finite subsets of $J$.  Note that $\mathcal{J}$ is a net with order given by set inclusion and join given by the union.

In particular, for our concerns, the metric topologies for $\LdpF{t}$ for $p \in [0,+\infty]$
are all Hausdorff commutative topological groups.  (If $p = 0$ then we consider convergence in measure, which is equivalent to a metric space with metric $d(f,g) = \int_{\Omega} \frac{|f - g|}{1 + |f-g|} d\P$ (lemma 13.40 in~\cite{AB07})).
\end{rem}

\begin{proof}[Proof of lemma~\ref{lemma_partition}]
Note that $1_{D} R_t(X) = \lrcurly{1_{D} u: u \in R_t(X)}$ for any $D \in \Ft{t}$. For notational convenience let $\hat{R}_t(X) := \lrcurly{u \in M_t: \exists J \subseteq I \mbox{ with }\P(\cup_{j \in J} A_j) = 1\mbox{ such that } 1_{A_j} u \in 1_{A_j} R_t(X)\; \forall j \in J} $.
\begin{enumerate}
\item The inclusion $R_t \subseteq \hat{R}_t$ follows straight forward: Let $u \in R_t(X)$, then by definition $1_D u \in 1_D R_t(X)$ for any $D \in \Ft{t}$, and in particular this is true for $D = A_i$ for any $i \in I$.  Therefore it follows that $u \in \hat{R}_t(X)$.
\item To prove $\hat{R}_t\subseteq R_t$ we will consider the two case: finite and infinite partitions. Let $u \in \hat{R}_t(X)$ and $J \subseteq I$ the underlying subindex.  Then $u = \sum_{j \in J} 1_{A_j} u$ almost surely, therefore $u \in R_t(X)$ if and only if $\sum_{j \in J} 1_{A_j} u \in R_t(X)$ since they are in the same equivalence class. Let $\#J$ denote the cardinality of the set $J$.  Note that by definition $1_{A_j} u \in 1_{A_j} R_t(X)$ for every $j \in J$.
    \begin{enumerate}
    \item If $\#J < +\infty$, i.e. if $J$ is a finite set, then trivially
        \[\sum_{j \in J} 1_{A_j} u \in \sum_{j \in J} 1_{A_j} R_t(X) \subseteq R_t(X)\]
        by closedness and conditional convexity of $R_t(X)$ as shown in proposition~\ref{prop_conditional_convex} below.  And thus $u \in R_t(X)$.
    \item Consider the case $\#J = +\infty$, i.e. if $J$ is not a finite set. Let $u \in \hat{R}_t(X)$, that is there exists $J \subseteq I$ with $\P(\cup_{j \in J} A_j) = 1$ such that $1_{A_j} u \in 1_{A_j} R_t(X)$ for all $j \in J$, or equivalently $1_{A_j} (u-m) \in 1_{A_j} R_t(X+m)$ for all $j \in J$ for some $m \in R_t(X)$ by using the translation property of $R_t$. We want to show $u \in R_t(X)$, respectively $u - m \in R_t(X + m)$.  Recall the summation as given in remark~\ref{rem_summ}, and the notation $\mathcal{J} = \lrcurly{K \subseteq J: \#K < +\infty}$.
        \begin{align}
        \nonumber u-m&=\sum_{j \in J} 1_{A_j} (u-m) \in \sum_{j \in J} 1_{A_j} R_t(X + m) = \lrcurly{\sum_{j \in J} 1_{A_j} Z_j: \forall j \in J: Z_j \in R_t(X + m)}\\
        \label{eq_partition-1} &= \lrcurly{\lim_{K \in \mathcal{J}} \sum_{k \in K} 1_{A_k} Z_k: \forall j \in J: Z_j \in R_t(X + m)}\\
        \nonumber &= \lrcurly{\lim_{K \in \mathcal{J}} \lrparen{\sum_{k \in K} 1_{A_k} Z_k + 1_{(\cup_{j \in J \backslash K} A_j)} 0}: \forall j \in J: Z_j \in R_t(X + m)}\\
        \label{eq_partition-2} &\subseteq \lrcurly{\lim_{K \in \mathcal{J}} \lrparen{\sum_{k \in K} 1_{A_k} Z_k + 1_{(\cup_{j \in J \backslash K} A_j)} \bar{Z}}: \forall j \in J: Z_j,\bar{Z} \in R_t(X + m)}\\
        \nonumber &\subseteq \liminf_{K \in \mathcal{J}} \lrcurly{\sum_{k \in K} 1_{A_k} Z_k + 1_{(\cup_{j \in J \backslash K} A_j)} \bar{Z}: \forall k \in K: Z_k,\bar{Z} \in R_t(X + m)}\\
        \nonumber &= \liminf_{K \in \mathcal{J}} \lrparen{\sum_{k \in K} 1_{A_k} R_t(X+m) + 1_{(\cup_{j \in J \backslash K} A_j)} R_t(X+m)}\\
        \label{eq_partition-3} &= \liminf_{K \in \mathcal{J}} R_t(X + m) = R_t(X+m).
        \end{align}
        Equation~\eqref{eq_partition-1} follows from the definition of an arbitrary summation as given in~\cite{B66,C66}, see remark~\ref{rem_summ}.  Inclusion~\eqref{eq_partition-2} follows from $0 \in R_t(X+m)$ since $m \in R_t(X)$.
        Equation~\eqref{eq_partition-3} follows from the finite case given above applied to the partition $((A_k)_{k \in K},\cup_{j \in J \backslash K} A_j)$.  Note that $\cup_{j \in J \backslash K} A_j \in \Ft{t}$ by $\seq{{\Ft{}}}$ a filtration satisfying the usual conditions (and $\Ft{t}$ is a sigma algebra).  Furthermore, note that we define the limit inferior as in~\cite{L11} to be $\liminf_{n \in N} B_n = \bigcap_{n \in N} \cl\bigcup_{m \geq n} B_m$ for a net of sets $(B_n)_{n \in N}$.
    \end{enumerate}
\end{enumerate}
\end{proof}

The following proposition is used in the proof of lemma~\ref{lemma_partition}.
\begin{prop}
\label{prop_conditional_convex}
A closed set $D \subseteq \LdpF{t}$ is conditionally convex if and only if for any $N \in \mathbb{N}$ where $N \geq 2$
\begin{equation}
\label{eq_N-convex}
\sum_{n = 1}^N \lambda_n D \subseteq D
\end{equation}
for every $(\lambda_n)_{n = 1}^N \in \Lambda_N := \lrcurly{(x_n)_{n = 1}^N: \sum_{n = 1}^N x_n = 1 \sas,  x_n \in L^\infty_t(\R_{+}) \;\forall n \in \{1,\ldots,N\}}$.
\end{prop}
\begin{proof}
\begin{enumerate}
\item[$\Leftarrow$] If $N = 2$ then this is the definition of conditional convexity.  If $N > 2$ then choose $(\lambda_n)_{n = 1}^N$
    such that $\lambda_n = 0$ almost surely for every $n > 2$, this then reduces to the case when $N = 2$ and thus $D$ is conditionally
    convex.

\item[$\Rightarrow$] We will first define a set of multipliers for strict convex combinations \[\Lambda_N^> = \lrcurly{(x_n)_{n = 1}^N: \sum_{n = 1}^N x_n = 1 \sas, x_n \in L^\infty_t(\R_{++})\;\forall n \in \{1,\ldots,N\}}.\]  Then the result for $\Lambda_N^>$ for any $N \in \mathbb{N}$ follows as in the static case (i.e. when $x_n\in\R_{++}$) by induction.

    Let $(\lambda_n)_{n = 1}^N \in \Lambda_N$.  Then there exists a sequence of $((\lambda_n^m)_{n = 1}^N)_{m = 0}^{+\infty} \subseteq \Lambda_N^>$ which converges almost surely to $(\lambda_n)_{n = 1}^N$ (i.e. for any $n \in \{1,\ldots,N\}$, $(\lambda_n^m)_{m = 0}^{+\infty}$ converges almost surely to $\lambda_n$, and for every $m$ it holds $\sum_{n = 1}^N \lambda_n^m = 1$ almost surely).  By the dominated convergence theorem, it follows that $\lambda_n^m X$ converges to $\lambda_n X$ in the metric topology
    for any $X \in \LdpF{t}$.  Therefore for any $(X_n)_{n = 1}^N \subseteq D$ (and let $\bar{X}_m = \sum_{n = 1}^N \lambda_n^m X_n \in D$ for any $m$)
    \begin{align*}
    \sum_{n = 1}^N \lambda_n X_n &= \sum_{n = 1}^N \lim_{m \to +\infty} \lambda_n^m X_n = \lim_{m \to +\infty} \sum_{n = 1}^N \lambda_n^m X_n = \lim_{m \to +\infty} \bar{X}_m \in D
    \end{align*}
    by $\bar{X}_m$ convergent (since it is the finite sum of converging series) and $D$ closed.
\end{enumerate}
\end{proof}

\begin{cor}
\label{cor_convex_decomposable}
Any conditional risk measure $R_t$ with closed and conditionally convex images has $\Ft{t}$-decomposable images.
\end{cor}
\begin{proof}
Let $R_t$ be a conditional risk measure with closed and conditionally convex images, and let $X \in \LdpF{}$.  Let $(\Omega_t^n)_{n =1}^N \subseteq \Ft{t}$, for some $N \in \mathbb{N}$, be a finite partition of $\Omega$.  By lemma~\ref{lemma_partition}, \[R_t(X) = \lrcurly{u \in M_t: \exists J \subseteq \{1,...,N\}: \P(\cup_{j \in J} \Omega_t^j) = 1, \forall j \in J: 1_{\Omega_t^j} u \in 1_{\Omega_t^j} R_t(X)}.\]  Therefore, if $(u_n)_{n =1}^N \subseteq R_t(X)$, 
then $1_{\Omega_t^m} \sum_{n =1}^N 1_{\Omega_t^n} u_n = 1_{\Omega_t^m} u_m \in 1_{\Omega_t^m} R_t(X)$ for every $m \in \{1,...,N\}$, and thus $\sum_{n =1}^N 1_{\Omega_t^n} u_n \in R_t(X)$.
\end{proof}

We showed that when the dynamic risk measure has closed and conditionally convex images, the set-optimization  approach of section~\ref{sec_vector}
and the measurable selector approach of section~\ref{sec_measurable} coincide. As a conclusion, the set-optimization  approach which is
using convex analysis results for set-valued functions, i.e. set-optimization, seems to be the richer approach as it allows to handle primal and dual representations for
$\LdpF{}$ spaces ($p\in[1,+\infty]$) as well as for the space $B_{K_T,n}$ (or any other locally convex preimage space). Furthermore, it allows to consider conditionally convex (and not necessarily conditionally coherent) risk measures as well as convex risk measures, whereas the measurable selectors approach relies heavily on the conditional coherency assumption.

\subsection{Set-optimization approach versus  set-valued portfolios}
\label{sec_vector-portfolio}
As in the prior sections, consider $\LdpF{}$ spaces with $p \in [0,+\infty]$.
\begin{thm}
\label{thm_vector-portfolio}
Given a conditional risk measure $R_t: \LdpF{} \to \mathcal{P}(M_t;M_{t,+})$ (see definition~\ref{defn_conditional}), then the function
${\bf R}_t: \mathcal{S}_T^d \to \mathcal{P}(M_t;M_{t,+})$ defined by
\begin{equation}
\label{ME}
{\bf R}_t(\X) := \bigcup_{Z \in \LdpK{p}{}{\X}} R_t(Z)
\end{equation}
for any set-valued portfolio $\X$ is a set-valued conditional risk measure (see definition~\ref{defn_fun-sets_rm}).

Given a set-valued conditional risk measure ${\bf R}_t: \bar{\mathcal {S}}_T^d \to \mathcal{P}(M_t;M_{t,+})$ (see
definition~\ref{defn_fun-sets_rm}) and a mapping $\X: \LdpF{} \to \bar{\mathcal {S}}_T^d$ of the set-valued portfolio associated with a
(random) portfolio vector such that $\X$ is monotone and translative, i.e. $\X(X) \subseteq \X(Y)$ if $Y - X \in \LdpF{+}$ and $\X(X + u) = \X(X) + u$ for any $X,Y \in \LdpF{}$ and $u \in M_t$, then the function
$R_t: \LdpF{} \to \mathcal{P}(M_t;M_{t,+})$ defined by
\begin{equation}
\label{RR}
R_t(X) := {\bf R}_t(\X(X))
\end{equation}
for any $X \in \LdpF{}$ is a conditional risk measure (see definition~\ref{defn_conditional}) which might not be finite at zero.
\end{thm}
\begin{proof}
\begin{enumerate}
\item Let $R_t: \LdpF{} \to \mathcal{P}(M_t;M_{t,+})$ be a conditional risk measure as in definition~\ref{defn_conditional}.  Let ${\bf R}_t(\X) := \bigcup_{Z \in \LdpK{p}{}{\X}} R_t(Z)$ for any set-valued portfolio $\X$.  We wish to show that ${\bf R}_t$ satisfies definition~\ref{defn_fun-sets_rm}.
    \begin{enumerate}
    \item Trivially ${\bf R}_t(\X) \in \mathcal{P}(M_t;M_{t,+})$ for any set-valued portfolio $\X$.
    \item Cash invariance: let $\X$ be a set-valued portfolio and let $m \in M_t$, then
        \begin{align*}
        {\bf R}_t(\X + m) &= \bigcup_{Z \in \LdpK{p}{}{\X + m}} R_t(Z) = \bigcup_{Z \in \LdpK{p}{}{\X}} R_t(Z +
m)\\
        &= \bigcup_{Z \in \LdpK{p}{}{\X}} R_t(Z) - m = {\bf R}_t(\X) - m.
        \end{align*}
    \item Monotonicity: Let $\X \subseteq {\bf Y}$ almost surely, then
        \begin{align*}
        {\bf R}_t(\X) &= \bigcup_{Z \in \LdpK{p}{}{\X}} R_t(Z) \subseteq \bigcup_{Z \in \LdpK{p}{}{{\bf Y}}} R_t(Z)
= {\bf R}_t({\bf Y}).
        \end{align*}
    \end{enumerate}
\item Let ${\bf R}_t: \bar{\mathcal {S}}_T^d \to \mathcal{P}(M_t;M_{t,+})$ be a set-valued conditional risk measure as in
definition~\ref{defn_fun-sets_rm}.  Let $\X: \LdpF{} \to\bar{\mathcal {S}}_T^d$ be a mapping of portfolio vectors to set-valued
portfolios that is monotone and translative. 
Let $R_t(X) := {\bf R}_t(\X(X))$ for any $X \in \LdpF{}$.  We wish to show that $R_t$ satisfies definition~\ref{defn_conditional}.
    \begin{enumerate}
    \item $\LdpF{+}$-monotonicity: Let $X,Y \in \LdpF{}$ such that $Y - X \in \LdpF{+}$.  Then $\X(Y) \supseteq \X(X)$, and thus
        $R_t(X) = {\bf R}_t(\X(X)) \subseteq {\bf R}_t(\X(Y)) = R_t(Y)$.
    \item $M_t$-translativity: Let $X \in \LdpF{}$ and $m \in M_t$, then
        \begin{align*}
        R_t(X + m) &= {\bf R}_t(\X(X + m)) = {\bf R}_t(\X(X) + m) = {\bf R}_t(\X(X)) - m = R_t(X) - m.
        \end{align*}
    \end{enumerate}
\end{enumerate}
\end{proof}

The above theorem states that conditional risk measures as in definition~\ref{defn_conditional} can be used to construct set-valued conditional risk measure (see definition~\ref{defn_fun-sets_rm}). This is in analogy to construction~\eqref{eq_constructive}, but yields a larger class of risk measures. If one restricts oneself to set-valued portfolios
$\X: \LdpF{} \to \bar{\mathcal{S}}_T^d$ which are monotonic and with $\X(X + m) = \X(X) + m$ for any $X \in \LdpF{}$ and $m \in M_t$, then conditional risk measures as in definition~\ref{defn_conditional} are one-to-one to set-valued conditional risk measure as in definition~\ref{defn_fun-sets_rm}.  This is the case whenever the set of portfolios $\X$ represents the set of portfolios that can be obtained from $X \in \LdpF{}$ following certain exchange rules (including transaction costs, trading constraints, illiquidity). The advantage of considering ${\bf R}_t$ as a function of the set $\X(X)$ as opposed to a function of $X$ as in \eqref{RR} is that ${\bf R}_t$ might be law invariant (see theorem~\ref{thm_constructive}), whereas $R_t$ is in general not law invariant.

\begin{ex}
\label{ex_set-portfolio_sum}
If $\X(X) := X + K$ for some (almost surely) closed convex lower set $K$ such that $\LdpK{p}{}{K}$ is closed, then trivially $\X(X)$ is a set-valued portfolio and satisfies monotonicity and translativity. 
\end{ex}

If $\X(X)$ is as in example~\ref{ex_set-portfolio_sum} and $K$ is additionally a convex cone, then for a given set-valued
conditional risk measure ${\bf R}_t$, the associated conditional risk measure $R_t$ defined by \eqref{RR} is $\LdpK{p}{}{K}$-compatible.

Note, that constructions very similar to \eqref{ME} appear a) in \cite{HRY12,AHR13} to define the market extension (that is a $C_{t,T}$-compatible version) of a risk measures $R_t$ by
\begin{equation*}
R_t^{mar}(X) := \bigcup_{Z \in X+C_{t,T}} R_t(Z),
\end{equation*}
where $C_{t,T}=-\sum_{s=t}^T \LdpK{p}{s}{K_s}$ and $(K_t)_{t=0}^T$ is a sequence of solvency cones modeling the bid-ask prices of the $d$ assets, and b) in \cite{FR12,FR12b} to define a multi-portfolio time consistent risk measure $(\tilde R_t)_{t=0}^T$ by backward recursion of a discrete time dynamic risk measure $(R_t)_{t=0}^T$ via  $\tilde R_T(X) = R_T(X)$ and
\begin{equation*}
\tilde R_t(X) := \bigcup_{Z \in \tilde R_{t+1}(X) } R_t(-Z)
\end{equation*}
for $t\in\{T-1,...,0\}$.

The following two corollaries provide additional relations between the conditional risk measures of the set-optimization approach and the set-valued portfolio conditional risk measures.  Specifically, they provide sufficient conditions for (conditional) convexity and coherence of one type of risk measure to be associated with a (conditionally) convex and coherent risk measure of the other type.
\begin{cor}
\label{cor_vector2portfolio}
Let $R_t: \LdpF{} \to \mathcal{P}(M_t;M_{t,+})$ be a convex (conditionally convex, positive homogeneous, conditionally positive homogeneous) conditional risk measure (see definition~\ref{defn_conditional}) at time $t$, then the associated set-valued conditional risk measure (see definition~\ref{defn_fun-sets_rm})
${\bf R}_t$ defined by \eqref{ME}  is convex (conditionally convex, positive homogeneous, conditionally positive homogeneous).
\end{cor}
\begin{proof}
Let $R_t: \LdpF{} \to \mathcal{P}(M_t;M_{t,+})$ be a conditional risk measure and let ${\bf R}_t(\X) := \bigcup_{Z \in \LdpK{p}{}{\X}} R_t(Z)$ for any $\X\in \bar{\mathcal {S}}_T^d$.
    \begin{enumerate}
    \item Let $R_t$ be convex.  Consider $\X,{\bf Y}\in \bar{\mathcal {S}}_T^d$ and $\lambda \in [0,1]$. Then,
        \begin{align*}
        {\bf R}_t(\lambda \X \oplus (1-\lambda) {\bf Y}) &= \bigcup_{Z \in \LdpK{p}{}{\lambda \X \oplus (1-\lambda){\bf Y}}} R_t(Z)\supseteq \bigcup_{Z \in \cl\lrparen{\lambda \LdpK{p}{}{\X} + (1-\lambda) \LdpK{p}{}{{\bf Y}}}} R_t(Z)\\
        &\supseteq \bigcup_{\substack{Z_X \in \LdpK{p}{}{\X}\\ Z_Y \in \LdpK{p}{}{{\bf Y}}}} R_t(\lambda Z_X + (1-\lambda) Z_Y)\\
        &\supseteq \bigcup_{\substack{Z_X \in \LdpK{p}{}{\X}\\ Z_Y \in \LdpK{p}{}{{\bf Y}}}} \lrsquare{\lambda R_t(Z_X) + (1-\lambda) R_t(Z_Y)}\\
        &= \lambda \bigcup_{Z_X \in \LdpK{p}{}{\X}} R_t(Z_X) + (1-\lambda) \bigcup_{Z_Y \in \LdpK{p}{}{{\bf Y}}} R_t(Z_Y)\\
        &= \lambda {\bf R}_t(\X) + (1-\lambda) {\bf R}_t({\bf Y}).
        \end{align*}
        The inclusion on the first line follows from $\cl\lrparen{\LdpK{p}{}{{\bf Z}_1} + \LdpK{p}{}{{\bf Z}_2}} \subseteq
            \LdpK{p}{}{{\bf Z}_1 \oplus {\bf Z}_2}$ for any random sets ${\bf Z}_1,{\bf Z}_2$ (with the
            norm topology on $p \in [1,+\infty]$, the metric topology on $p \in (0,1)$, and the topology generated by convergence in probability for $p = 0$);             
            for $p \in [1,+\infty)$ equality holds.
    \item Let $R_t$ be conditionally convex.  Then the proof is analogous to the convex case above.
    \item Let $R_t$ be positive homogeneous.  Consider $\X\in \bar{\mathcal {S}}_T^d$ and $\lambda > 0$. It holds
        \begin{align*}
        {\bf R}_t(\lambda \X) &= \bigcup_{Z \in \LdpK{p}{}{\lambda \X}} R_t(Z) = \bigcup_{Z \in \LdpK{p}{}{\X}} R_t(\lambda Z) = \lambda \bigcup_{Z \in \LdpK{p}{}{\X}} R_t(Z) = \lambda {\bf R}_t(\X).
        \end{align*}
    \item Let $R_t$ be conditionally positive homogeneous.  Then the proof is analogous to the positive homogeneous case above.
    \end{enumerate}
\end{proof}

\begin{cor}
\label{cor_portfolio2vector}
Let ${\bf R}_t: \bar{\mathcal {S}}_T^d \to \mathcal{P}(M_t;M_{t,+})$ be a set-valued conditional risk measure (see definition~\ref{defn_fun-sets_rm}) at time $t$, and let $\X:
\LdpF{} \to \bar{\mathcal {S}}_T^d$ of the set-valued portfolio associated with a (random) portfolio vector be monotonic and translative.
Let $R_t$ be the associated conditional risk measure (see definition~\ref{defn_conditional}).
\begin{enumerate}
\item If ${\bf R}_t$ is convex and $\X(\lambda X + (1-\lambda) Y) \supseteq \lambda \X(X) \oplus (1-\lambda)\X(Y)$
for every $X,Y \in \LdpF{}$ and $\lambda \in [0,1]$ ($\X$ is closed-convex), then $R_t$ is convex.
\item If ${\bf R}_t$ is conditionally convex and $\X(\lambda X + (1-\lambda) Y) \supseteq \lambda \X(X) \oplus (1-\lambda)\X(Y)$
for every $X,Y \in \LdpF{}$ and $\lambda \in \LiF{t}$ with $0 \leq \lambda \leq 1$ ($\X$ is conditionally closed-convex), then $R_t$ is conditionally convex.
\item If ${\bf R}_t$ is positive homogeneous and $\X(\lambda X) = \lambda \X(X)$ for every $X \in \LdpF{}$ and $\lambda >
0$ ($\X$ is positive homogeneous), then $R_t$ is positive homogeneous.
\item If ${\bf R}_t$ is conditionally positive homogeneous and $\X(\lambda X) = \lambda \X(X)$ for every $X \in \LdpF{}$
and $\lambda \in L^\infty_t(\R_{++})$ ($\X$ is conditionally positive homogeneous), then $R_t$ is conditionally positive homogeneous.
\end{enumerate}
\end{cor}
\begin{proof}
Let ${\bf R}_t: \bar{\mathcal {S}}_T^d \to \mathcal{P}(M_t;M_{t,+})$ be a set-valued conditional risk measure, let $\X$
be as above and let $R_t(X) := {\bf R}_t(\X(X))$ for every portfolio vector $X \in \LdpF{}$.
    \begin{enumerate}
    \item Let ${\bf R}_t$ be convex and $\X$ be closed-convex.  Let $X,Y \in \LdpF{}$ and $\lambda \in [0,1]$.
        \begin{align*}
        R_t(\lambda X + (1-\lambda)Y) &= {\bf R}_t(\X(\lambda X + (1-\lambda)Y)) \supseteq {\bf R}_t(\lambda \X(X) \oplus (1-\lambda)\X(Y))\\
        &\supseteq \lambda {\bf R}_t(\X(X)) + (1-\lambda) {\bf R}_t(\X(Y)) = \lambda R_t(X) + (1-\lambda) R_t(Y).
        \end{align*}
    \item Let ${\bf R}_t$ be conditionally convex and $\X$ be conditionally closed-convex.  Then the proof is analogous to the convex case above.
    \item Let ${\bf R}_t$ and $\X$ be positive homogeneous.  Let $X \in \LdpF{}$ and $\lambda > 0$.
        \begin{align*}
        R_t(\lambda X) &= {\bf R}_t(\X(\lambda X)) = {\bf R}_t(\lambda \X(X)) = \lambda {\bf R}_t(\X(X)) = \lambda R_t(X).
        \end{align*}
    \item Let ${\bf R}_t$ and $\X$ be conditionally positive homogeneous.  Then the proof is analogous to the positive homogeneous case above.
    \end{enumerate}
\end{proof}

\begin{ex}(Example~\ref{ex_set-portfolio_sum} continued)
\label{ex_set-portfolio_sum2}
Let $\X(X) := X + K$ for every $X \in \LdpF{}$ for some random set $K$.  If $K$ is (almost surely) convex and closed then $\X$
is ($\Ft{}$-)conditionally closed-convex (and thus closed-convex as well).  If $K$ is (almost surely) a cone then $\X$ is ($\Ft{}$-)conditionally positive
homogeneous (and thus positive homogeneous as well).
\end{ex}

In light of theorem~\ref{thm_vector-portfolio}, equation~\eqref{RR} and corollary~\ref{cor_portfolio2vector} for set-valued portfolios of the form $\X(X) := X + K$ for all $X \in \LdpF{}$ and some random closed convex cone $K$, one obtains the following. The dual representation of a constructive risk measure ${\bf R}_0$ with coherent components $\rho^1,...,\rho^d$ given in equation~(5.2) in \cite{CM13} coincides with a special case of the dual representation of a $K_T$-compatible risk measure $R_0$ given in Theorem~4.2 in \cite{HHR10}, by choosing $A=\times_{i=1}^d A_i$ ($A_i$ being the acceptance set of $\rho_i$), $M_0=R^d$, $K_I=R^d_+$ and $K_T=-K$:
\[
R_0(X)={\bf R}_0(X+K)=\bigcap_{w\in \R^d_+\backslash\{0\},\Q\in\mathcal Q, \diag{w}\dQdP \in \plusp{-K} \mbox{ a.s.}}\{u\in\R^d: \trans{w}\EQ{X}\leq \trans{w}u\},
\]
where $\mathcal Q=\times_{i=1}^d \mathcal Q_i$ and $\mathcal Q_i$ denotes the set of probability measures in the dual
representation of $\rho_i$. This also follows from corollary~\ref{cor_conditional_dual}, where the set of dual variables is
\[
\W_0^{\max} = \lrcurly{(\Q,w) \in \W_0: w_t^T(\Q,w) \in \plus{A_t}}=\lrcurly{(\Q,w) \in \W_0:\Q \in \mathcal Q},
\]
with
\begin{equation*}
\W_0 := \lrcurly{(\Q,w) \in \mathcal{M} \times \R^d_+\backslash\{0\}: w_0^T(\Q,w) \in L^q_d(\mathcal F_T;K_T^+)}
\end{equation*}
due do $K_T$-compatibility of $R_0$.

Additional to dual representations for constructive risk measure, theorem~\ref{thm_vector-portfolio} allows to deduce
dual representations of a larger class of conditional risk measure for set-valued portfolios (definition~\ref{defn_fun-sets_rm}) by using equation~\eqref{ME} and the duality results for set-valued risk measures of the set-optimization approach.

\subsection{Set-optimization approach versus family of scalar risk measures}
\label{sec_vector-scalar}
For this section consider $p \in [1,+\infty]$, where $\LdpF{t}$ has the norm topology for any $p \in [1,+\infty)$ and the weak* topology for $p = +\infty$.  In the static setting, the relation between set-valued risk measures and multiple asset scalar risk measures has been studied in~\cite{HH10,HHR10,FMM13}.

\begin{thm}
\label{thm_vector_scalar}
Let $R_t: \LdpF{} \to \mathcal{P}(M_t;M_{t,+})$ be a conditional risk measure at time $t$ (see definition~\ref{defn_conditional}), then $\rho_t^{w}: \LdpF{} \to L^0_t(\bar\R)$, defined by \[\rho_t^{w}(X) := \essinf_{u \in R_t(X)} \trans{w}u\] for any $X \in \LdpF{}$, is a family of multiple asset scalar risk measures indexed by $w \in \plus{M_{t,+}} \backslash \prp{M_t}$ at time $t$ (see definition~\ref{defn_multi-scalar}).

Let $\lrcurly{\rho_t^{w}: \LdpF{} \to L^0_t(\bar\R): w \in \plus{M_{t,+}} \backslash \prp{M_t}}$ be a family of multiple asset scalar
risk measures at time $t$ indexed by $w \in \plus{M_{t,+}} \backslash \prp{M_t}$ (see definition~\ref{defn_multi-scalar}), then $R_t:
\LdpF{} \to \mathcal{P}(M_t;M_{t,+})$, defined by \[R_t(X) := \bigcap_{w \in \plus{M_{t,+}} \backslash \prp{M_t}} \lrcurly{u \in M_t:
\rho_t^{w}(X) \leq \trans{w}u \Pas}\] for any $X \in \LdpF{}$, is a conditional risk measure at time $t$ (see definition~\ref{defn_conditional}).
\end{thm}
\begin{proof}
\begin{enumerate}
\item This follows form definition~\ref{defn_multi-scalar} and \eqref{scalarization}.
\item We will show that $R_t(X) := \bigcap_{w \in \plus{M_{t,+}} \backslash \prp{M_t}} \lrcurly{u \in M_t: \rho_t^{w}(X) \leq \trans{w}u \Pas}$ is a conditional risk measure.  We use the properties of $\rho_t^{w}$ given in proposition~\ref{prop_multi-scalar_properties}.
    \begin{enumerate}
    \item $\LdpF{+}$-monotonicity: let $X,Y \in \LdpF{}$ such that $Y - X \in \LdpF{+}$, then $\rho_t^{w}(Y) \leq \rho_t^{w}(X)$ almost surely for every $w \in \plus{M_{t,+}} \backslash \prp{M_t}$.  Therefore $R_t(Y) \supseteq R_t(X)$.
    \item $M_t$-translativity: let $X \in \LdpF{}$ and $m \in M_t$, then
        \begin{align*}
        R_t(X + m) &= \bigcap_{w \in \plus{M_{t,+}} \backslash \prp{M_t}} \lrcurly{u \in M_t: \rho_t^{w}(X + m) \leq \trans{w}u \Pas}\\
        &= \bigcap_{w \in \plus{M_{t,+}} \backslash \prp{M_t}} \lrcurly{u \in M_t: \rho_t^{w}(X) - \trans{w}m \leq \trans{w}u \Pas}\\
        &=\bigcap_{w \in \plus{M_{t,+}} \backslash \prp{M_t}} \lrcurly{u \in M_t: \rho_t^{w}(X) \leq \trans{w}(u+m) \Pas}\\
        &= \bigcap_{w \in \plus{M_{t,+}} \backslash \prp{M_t}} \lrcurly{u \in M_t: \rho_t^{w}(X) \leq \trans{w}u \Pas} - m = R_t(X) - m.
        \end{align*}
    \item Finiteness at zero: $R_t(0) \neq \emptyset$ since $\rho_t^{w}(0) < +\infty$ for every $w \in \plus{M_{t,+}} \backslash \prp{M_t}$, and $R_t(0)[\omega] \neq \tilde{M}_t[\omega]$ since there exists a $v \in \plus{M_{t,+}} \backslash \prp{M_t}$ such that $\rho_t^{v}(0) > -\infty$.
    \end{enumerate}
\end{enumerate}
\end{proof}

\begin{rem}
\label{rem_scalarization}
If $R_t$ is normalized, with closed and conditionally convex images, and $w \in \plus{R_t(0)} \backslash \prp{M_t}$ then $\rho_t^{w}(0)= 0$, i.e. $\rho_t^{w}$ normalized in the scalar framework.
\end{rem}

Apart from closedness, many properties are one-to-one for conditional risk measures $R_t$ and the corresponding family of scalarizations. The corresponding results for the static case can be found in lemma 5.1 and lemma 6.1 of~\cite{HH10}. An example showing that closedness of $R_t$ does not necessarily imply closedness of all scalarizations can be found in the beginning of section~5 in~\cite{HH10} for the case $t=0$.

\begin{cor}
\label{cor_vector_scalar}
Let $R_t: \LdpF{} \to \mathcal{P}(M_t;M_{t,+})$ be a convex (conditionally convex, positive homogeneous, conditionally positive
homogeneous) conditional risk measure at time $t$ with closed and $\Ft{t}$-decomposable images, then the associated family of
multiple asset scalar risk measures is convex (conditionally convex, positive homogeneous, conditionally positive homogeneous).

Let $\lrcurly{\rho_t^{w}: \LdpF{} \to L^0_t(\bar\R): w \in \plus{M_{t,+}} \backslash \prp{M_t}}$ be a family of convex (positive homogeneous, conditionally positive homogeneous, lower semicontinuous) multiple asset scalar risk measures at time $t$ indexed by $w \in \plus{M_{t,+}} \backslash \prp{M_t}$ then the associated conditional risk measure is convex (positive homogeneous, conditionally positive homogeneous, closed).  Additionally, if $\lrcurly{\rho_t^{w}: \LdpF{} \to L^0_t(\bar\R): w \in \plus{M_{t,+}} \backslash \prp{M_t}}$ is a family of lower semicontinuous conditionally convex risk measures then the associated conditional risk measure is conditionally convex.
\end{cor}
\begin{proof}
\begin{enumerate}
\item Let $R_t: \LdpF{} \to \mathcal{P}(M_t;M_{t,+})$ be a conditional risk measure at time $t$.  Let $\rho_t^{w}: \LdpF{} \to L^0_t(\bar\R)$ be defined by $\rho_t^{w}(X) := \essinf_{u \in R_t(X)} \trans{w}u$ for every $X \in \LdpF{}$.
    \begin{enumerate}
    \item Let $R_t$ be convex.  Let $X,Y \in \LdpF{}$, $\lambda \in [0,1]$, and $w \in \plus{M_{t,+}} \backslash \prp{M_t}$.
        \begin{align*}
        \rho_t^{w}(\lambda X + (1-\lambda) Y) &= \essinf_{u \in R_t(\lambda X + (1-\lambda) Y)} \trans{w}u\\
        &\leq \essinf_{u \in \lambda R_t(X) + (1-\lambda) R_t(Y)} \trans{w}u\\
        &= \lambda \essinf_{u_X \in R_t(X)} \trans{w}u_X + (1-\lambda) \essinf_{u_Y \in R_t(Y)} \trans{w}u_Y\\
        &= \lambda \rho_t^{w}(X) + (1-\lambda) \rho_t^{w}(Y).
        \end{align*}
    \item Let $R_t$ be conditionally convex.  Then the proof is analogous to the convex case above.
    \item Let $R_t$ be positive homogeneous.  Let $X \in \LdpF{}$, $\lambda > 0$, and $w \in \plus{M_{t,+}} \backslash \prp{M_t}$.
        \begin{align*}
        \rho_t^{w}(\lambda X) &= \essinf_{u \in R_t(\lambda X)} \trans{w}u = \essinf_{u \in \lambda R_t(X)} \trans{w}u = \lambda \essinf_{u \in R_t(X)} \trans{w}u = \lambda \rho_t^{w}(X).
        \end{align*}
    \item Let $R_t$ be conditionally positive homogeneous.  Then the proof  is analogous to the positive homogeneous case above.
    \end{enumerate}
\item Let $\lrcurly{\rho_t^{w}: \LdpF{} \to L^0_t(\bar\R): w \in \plus{M_{t,+}} \backslash \prp{M_t}}$ be a family of multiple asset scalar risk measures at time $t$ indexed by $w \in \plus{M_{t,+}} \backslash \prp{M_t}$.  Let $R_t: \LdpF{} \to \mathcal{P}(M_t;M_{t,+})$ be defined by $R_t(X) := \bigcap_{w \in \plus{M_{t,+}} \backslash \prp{M_t}} \lrcurly{u \in M_t: \rho_t^{w}(X) \leq \trans{w}u \Pas}$ for every $X \in \LdpF{}$.
    \begin{enumerate}
    \item Let $\rho_t^{w}$ be convex for every $w \in \plus{M_{t,+}} \backslash \prp{M_t}$.  Let $X,Y \in \LdpF{}$ and $\lambda \in (0,1)$.
        \begin{align*}
        R_t(\lambda X + (1-\lambda) Y) &= \bigcap_{w \in \plus{M_{t,+}} \backslash \prp{M_t}} \lrcurly{u \in M_t: \rho_t^{w}(\lambda X + (1-\lambda) Y) \leq \trans{w}u \Pas}\\
        &\supseteq \bigcap_{w \in \plus{M_{t,+}} \backslash \prp{M_t}} \lrcurly{u \in M_t: \lambda \rho_t^{w}(X) + (1-\lambda) \rho_t^{w}(Y) \leq \trans{w}u \Pas}\\
        &\supseteq \bigcap_{w \in \plus{M_{t,+}} \backslash \prp{M_t}} \lsquare{\lrcurly{\lambda u_X: u_X \in M_t, \rho_t^{w}(X) \leq \trans{w}u_X \Pas}}\\
        &\quad\quad \rsquare{+ \lrcurly{(1-\lambda) u_Y: u_Y \in M_t, \rho_t^{w}(Y) \leq \trans{w}u_Y \Pas}}\\
        &\supseteq \lambda \bigcap_{w \in \plus{M_{t,+}} \backslash \prp{M_t}} \lrcurly{u_X \in M_t: \rho_t^{w}(X) \leq \trans{w}u_X \Pas}\\
        &\quad\quad + (1-\lambda) \bigcap_{w \in \plus{M_{t,+}} \backslash \prp{M_t}} \lrcurly{u_Y \in M_t: \rho_t^{w}(Y) \leq \trans{w}u_Y \Pas}\\
        &= \lambda R_t(X) + (1-\lambda) R_t(Y).
        \end{align*}
        Let $\lambda = 0$ (the case for $\lambda = 1$ is analogous), then $R_t(\lambda X + (1-\lambda) Y) = \lambda R_t(X) + (1-\lambda) R_t(Y)$ for any conditional risk measure and the result follows.

    \item Let $\rho_t^{w}$ be positive homogeneous for every $w \in \plus{M_{t,+}} \backslash \prp{M_t}$.  Let $X \in \LdpF{}$ and $\lambda > 0$.
        \begin{align*}
        R_t(\lambda X) &= \bigcap_{w \in \plus{M_{t,+}} \backslash \prp{M_t}} \lrcurly{u \in M_t: \rho_t^{w}(\lambda X) \leq \trans{w}u \Pas}\\
        &= \bigcap_{w \in \plus{M_{t,+}} \backslash \prp{M_t}} \lrcurly{u \in M_t: \lambda \rho_t^{w}(X) \leq \trans{w}u \Pas}\\
        &= \bigcap_{w \in \plus{M_{t,+}} \backslash \prp{M_t}} \lrcurly{\lambda u: u \in M_t, \rho_t^{w}(X) \leq \trans{w}u \Pas}\\
        &= \lambda R_t(X).
        \end{align*}
    \item Let $\rho_t^{w}$ be conditionally positive homogeneous for every $w \in \plus{M_{t,+}} \backslash \prp{M_t}$.  Then the proof  is analogous to the positive homogeneous case above.
    \item Let $\rho_t^{w}$ be lower semicontinuous for every $w \in \plus{M_{t,+}} \backslash \prp{M_t}$. Consider a sequence $(X_n,u_n)_{n \in \mathbb{N}} \subseteq \operatorname{graph} R_t$ (respectively a net if $p = +\infty$) with $\lim_{n \to +\infty} (X_n,u_n) = (X,u)$.
    Note that $(X_n,u_n) \in \operatorname{graph} R_t$ if and only if $\rho_t^{v}(X_n) \leq \trans{v}u_n$ for every $v \in \plus{M_{t,+}} \backslash \prp{M_t}$.
        \begin{align*}
        \rho_t^{w}(X) \leq \liminf_{n \to +\infty} \rho_t^{w}(X_n) \leq \liminf_{n \to +\infty} \trans{w}u_n = \trans{w}u.
        \end{align*}
        The last equality above follows from $u_n \to u$ in $\LdpF{t}$ implies $\trans{w}u_n \to \trans{w}u$ in $\LoF{t}$ (by H\"{o}lder's inequality). Thus, $(X,u) \in \operatorname{graph} R_t$.
    \item Let $\rho_t^{w}$ be lower semicontinuous and conditionally convex for every $w \in \plus{M_{t,+}} \backslash \prp{M_t}$.  Let $X,Y \in \LdpF{}$ and $\lambda \in \LiF{t}$ with $0 < \lambda < 1$, then the proof is analogous to the convex case above.

        We will now extend conditional convexity to the case for $\lambda \in \LiF{t}$ with $0 \leq \lambda \leq 1$ in the same way as was
accomplished in the proof of~\cite[corollary 4.9]{FR12}, noting that $R_t$ is closed by $\rho_t^{w}$ lower semicontinuous.  Take a
sequence $(\lambda_n)_{n = 0}^{+\infty} \subseteq \LiF{t}$ such that $0 < \lambda_n < 1$ for every $n \in \mathbb{N}$ which converges
almost surely to $\lambda$.  Then by dominated convergence $\lambda_n X$ converges to $\lambda X$ in the norm topology
(weak* topology if $p = +\infty$) for any $X \in \LdpF{}$.  Therefore, for any $X,Y \in \LdpF{}$
        \begin{align*}
        R_t(\lambda X + (1-\lambda) Y) &= R_t(\lim_{n \to +\infty} (\lambda_n X + (1-\lambda_n) Y))\\
        &\supseteq \liminf_{n \to +\infty} R_t(\lambda_n X + (1-\lambda_n) Y)\\
        &\supseteq \liminf_{n \to +\infty} [\lambda_n R_t(X) + (1-\lambda_n) R_t(Y)]\\
        &\supseteq \lambda R_t(X) + (1-\lambda) R_t(Y)
        \end{align*}
        by $R_t$ closed (see proposition~2.34 in \cite{L11}) and conditionally convex on the interval $0 < \lambda_n < 1$. Note that we use the convention from~\cite{L11} that the limit inferior of a sequence of sets $(B_i)_{i \in \mathbb{N}}$ is given by $\liminf_{i \to +\infty} B_i = \bigcap_{i\in \mathbb{N}} \cl\bigcup_{j \geq i} B_j$.
    \end{enumerate}
\end{enumerate}
\end{proof}

In the following lemma we will show that when the conditional risk measure has closed and conditionally convex images, the family of
scalarizations can be used to recover the conditional risk measure.

\begin{lemma}
\label{lemma_scalar_rep}
Let $R_t: \LdpF{} \to \mathcal{P}(M_t;M_{t,+})$ be a dynamic risk measure with closed and conditionally convex images.  Then, for any $X \in \LdpF{}$
\begin{equation}
\label{eq_scalar_rep}
R_t(X) = \bigcap_{w \in \plus{M_{t,+}} \backslash \prp{M_t}} \lrcurly{u \in M_t: \rho_t^{w}(X) \leq \trans{w}u \Pas}
\end{equation}
where $\rho_t^{w}(X) := \essinf_{u \in R_t(X)} \trans{w}u$ is the multiple asset scalar risk measure associated with $R_t$.
\end{lemma}
\begin{proof}
\begin{enumerate}
\item[$\subseteq$:] By definition it is easy to see that $u \in R_t(X)$ implies that $\trans{w}u \geq \rho_t^{w}(X)$ for any $w \in \plus{M_{t,+}} \backslash \prp{M_t}$.
\item[$\supseteq$:] Let $u \in \bigcap_{w \in \plus{M_{t,+}} \backslash \prp{M_t}} \lrcurly{u \in M_t: \rho_t^{w}(X) \leq \trans{w}u \Pas}$.  Assume $u \not\in R_t(X)$.  Then since $R_t(X)$ is closed and convex we can apply the separating hyperplane theorem.  In particular, there exists a $v \in \plus{M_{t,+}} \backslash \prp{M_t}$ such that $\E{\trans{v}u} < \inf_{\hat{u} \in R_t(X)} \E{\trans{v}\hat{u}}$ (if $v \not\in \plus{M_{t,+}} \backslash \prp{M_t}$ then $\inf_{\hat{u} \in R_t(X)} \E{\trans{v}\hat{u}} = -\infty$ by $R_t(X) = R_t(X) + M_{t,+}$).  This implies that $\E{\rho_t^{v}(X)} = \E{\essinf_{\hat{u} \in R_t(X)} \trans{v}\hat{u}} \leq \E{\trans{v}u} < \inf_{\hat{u} \in R_t(X)} \E{\trans{v}\hat{u}}$.

    By corollary~\ref{cor_convex_decomposable}, $R_t(X)$ is $\Ft{t}$-decomposable.  Therefore by theorem 1 of~\cite{Y85} (and $\lrcurly{\trans{v}u: u \in R_t(X)} \subseteq \LoF{t}$), it follows that $\E{\essinf_{\hat{u} \in R_t(X)} \trans{v}\hat{u}} = \inf_{\hat{u} \in R_t(X)} \E{\trans{v}\hat{u}}$.  This is a contradiction and thus $u \in R_t(X)$.
\end{enumerate}
\end{proof}

\bibliographystyle{plain}
\bibliography{biblio}
\end{document}